\documentclass{article}
\usepackage[utf8]{inputenc}
\usepackage[english]{babel}
\usepackage[T1]{fontenc}
\usepackage[letterpaper,margin=1.00in]{geometry}
\usepackage{amsmath, amssymb, amsthm, amsfonts}
\usepackage{bbm}
\usepackage{amscd}
\usepackage{mathrsfs}

\usepackage{comment} 
\usepackage{ifthen}
\usepackage{tikz}
\usetikzlibrary{positioning,decorations.pathreplacing}
\usepackage{ bbold }
\usepackage{appendix}
\usepackage{graphicx}
\usepackage{color}
\usepackage{algorithm}
\usepackage[noend]{algpseudocode}
\usepackage{epstopdf}
\usepackage{wrapfig}
\usepackage{paralist}
\usepackage{wasysym}
\usepackage[disable]{todonotes}

\usepackage{listings} 

\usepackage{pdflscape} 

\usepackage{framed}
\usepackage[framemethod=tikz]{mdframed}
\usepackage[bottom]{footmisc}
\usepackage{enumitem}
\setitemize{noitemsep,topsep=3pt,parsep=3pt,partopsep=3pt}
\usepackage[font=small]{caption}
\usepackage{xspace}

\usepackage{thmtools} 
\usepackage{thm-restate} 

\usepackage{hyperref}
\hypersetup{
    unicode=false,          
    colorlinks=true,        
    linkcolor=red,          
    citecolor=darkgreen,        
    filecolor=magenta,      
    urlcolor=cyan           
}

\newtheorem{theorem}{Theorem}[section]
\newtheorem{lemma}[theorem]{Lemma}
\newtheorem{meta-theorem}[theorem]{Meta-Theorem}

\newtheorem{definition}[theorem]{Definition}

\usepackage[capitalize, nameinlink,noabbrev]{cleveref}

\crefname{theorem}{Theorem}{Theorems}
\crefname{proposition}{Proposition}{Propositions}
\crefname{observation}{Observation}{Observations}
\crefname{lemma}{Lemma}{Lemmas}
\crefname{claim}{Claim}{Claims}
\crefname{problem}{Problem}{Problems}
\crefname{conjecture}{Conjecture}{Conjectures}
\crefname{question}{Question}{Questions}
\crefname{example}{Example}{Examples}
\crefname{fact}{Fact}{Facts}

\definecolor{darkgreen}{rgb}{0,0.5,0}

\algnewcommand\algorithmicswitch{\textbf{switch}}
\algnewcommand\algorithmiccase{\textbf{case}}

\algdef{SE}[SWITCH]{Switch}{EndSwitch}[1]{\algorithmicswitch\ #1\ \algorithmicdo}{\algorithmicend\ \algorithmicswitch}%
\algdef{SE}[CASE]{Case}{EndCase}[1]{\algorithmiccase\ #1}{\algorithmicend\ \algorithmiccase}%
\algtext*{EndSwitch}%
\algtext*{EndCase}%

\newcommand{\eps}{\varepsilon}

\newcommand{\ksi}{\xi}

\newcommand{\polylog}{\operatorname{polylog}}

\newcommand{\poly}{\operatorname{poly}}

\renewcommand{\phi}{\varphi}

\renewcommand{\paragraph}[1]{\vspace{0.15cm}\noindent {\bf #1}:}

\newcommand{\FullOrShort}{full}
\ifthenelse{\equal{\FullOrShort}{full}}{
  
  \newcommand{\fullOnly}[1]{#1}
  \newcommand{\shortOnly}[1]{}

  }{

    \newcommand{\fullOnly}[1]{}
    \newcommand{\IncludePictures}[1]{}
   
  }

\title{Online Consistent $k$-center with Constant Worst-Case Recourse}
\title{Low-Recourse Online $k$-centers with Deletions }
\title{Online k-Center with Constant Recourse}
\title{Fully Dynamic Consistent $k$-Center Clustering}

\author{
  Jakub Łącki \\
  \small{Google}\\
  \small{jlacki@google.com}
  \and \textcircled{r}\footnote{The author ordering was randomized using \url{https://www.aeaweb.org/journals/policies/random-author-order/generator}. 
 It is requested that citations of this work list the authors separated by \texttt{\textbackslash textcircled\{r\}} instead of commas. } 
  \and
  Bernhard Haeupler\thanks{Supported in part by NSF grants CCF-1814603, CCF-1910588, NSF CAREER award CCF-1750808, a Sloan Research Fellowship, and the European Union's Horizon 2020 ERC grant 949272.} \\
  \small{ETH Zurich and CMU}\\
  \small{bernhard.haeupler@inf.ethz.ch}
  \and \textcircled{r} \and  
  Christoph Grunau\thanks{Supported by the European Research Council (ERC) under the European Unions Horizon 2020 research and innovation programme (grant agreement No.~853109).} \\
  \small{ETH Zurich}\\
  \small{cgrunau@inf.ethz.ch}
  \and
 Václav Rozhoň \footnotemark[3]\\ 
\small{ETH Zurich} \\ 
\small{rozhonv@ethz.ch}
  \and \textcircled{r} \and 
  Rajesh Jayaram \\
  \small{Google}\\
  \small{rkjayaram@google.com}
}

\date{July 13, 2023}
\begin{document}

\maketitle

\begin{abstract}

We study the \emph{consistent k-center clustering} problem. In this problem, the goal is to maintain a constant factor approximate $k$-center solution during a sequence of $n$ point insertions and deletions while minimizing the \emph{recourse}, i.e., the number of changes made to the set of centers after each point insertion or deletion. Previous works by Lattanzi and Vassilvitskii [ICML '12] and Fichtenberger, Lattanzi, Norouzi-Fard, and Svensson [SODA '21] showed that in the incremental setting, where deletions are not allowed, one can obtain $k \cdot \polylog(n) / n$ amortized recourse for both $k$-center and $k$-median, and demonstrated a matching lower bound. However, no algorithm for the fully dynamic setting achieves less than the trivial $O(k)$ changes per update, which can be obtained by simply reclustering the full dataset after every update.

In this work, we give the first algorithm for consistent $k$-center clustering for the fully dynamic setting, i.e., when both point insertions and deletions are allowed, and improves upon a trivial $O(k)$ recourse bound. Specifically, our algorithm maintains a constant factor approximate solution while ensuring worst-case \emph{constant} recourse per update, which is optimal in the fully dynamic setting. Moreover, our algorithm is deterministic and is therefore correct even if an adaptive adversary chooses the insertions and deletions.
\end{abstract}

\thispagestyle{empty}
\newpage
\parskip 7.2pt 
\pagenumbering{arabic}

\section{Introduction}

Clustering is a fundamental problem in computer science,
which arises in approximation algorithms, unsupervised learning, computational
geometry, spam and community detection, image segmentation,
databases, and other areas \cite{hansen1997cluster,shi2000normalized,arthur2006k,schaeffer2007graph,fortunato2010community,coates2012learning,tan2013data}. The goal of clustering is to find a structure in data by grouping
together similar data points. Specifically, for $k$-clustering objectives, the goal is to output a set of $k$ cluster \emph{centers} from the underlying metric space; the cost of the clustering is then evaluated by various functions of the distances of each point to their nearest center. Due to the importance of the problem, clustering has been extensively studied and many algorithms \cite{
arya2001local, kanungo2002local, jain2003greedy,  charikar2005improved, arthur2007k, ahmadian2019better, byrka2017improved}
and heuristics \cite{lloyd1982least} have been proposed to solve the offline version of the
problem, where the data is fixed and provided to the algorithm up front. 


Recently, due to the proliferation
of enormous datasets and the rise of modern computational paradigms where data is constantly changing, there has been significant interest in developing
\emph{dynamic} clustering algorithms~\cite{cohen2016diameter,lattanzi2017consistent,ChaFul18,DBLP:conf/esa/GoranciHL18,schmidt2019fully,goranci2019fully,HenzingerK20,henzinger2020dynamic,fichtenberger2021consistent}, which can maintain a clustering of a set of points $P$ subject to updates.
The two primary settings in which these problems are studied are the \emph{incremental} setting, where $P$ undergoes solely point insertions, and the \textit{fully dynamic} setting, where points can be both inserted to and deleted from $P$.
The fully dynamic setting is strictly more general and  captures a much wider class of applications. 

In this paper, we focus on dynamic \emph{consistent} algorithms, following a recent line of work~\cite{lattanzi2017consistent,cohen2019fully,jaghargh2019consistent,fichtenberger2021consistent}.
Specifically, a dynamic clustering algorithm is said to be \emph{consistent} if the number of changes it makes to its set of $k$ cluster centers is small. In particular, it should be smaller than the trivial $k$ changes obtained from independently reclustering the data after every step. The number of changes made to the cluster centers after each point insertion or deletion is known as the \emph{recourse}\footnote{To be more precise, we define the recourse as the number of swaps needed, that is, when $|C_1|=|C_2|=k$ and $|C_1 \triangle C_2| = 2$, we say that the recourse is one. }.
We note that a different commonly studied objective in the area of dynamic algorithms, including algorithms for clustering problems, is minimizing the running time of each insertion or deletion operation \cite{charikar2004incremental,ChaFul18,schmidt2019fully,henzinger2020fully,goranci2019fully,cohen2019fully, bateni2023optimal}.

Consistency is an important measure from both a theoretical and a practical point of view. From a theoretical perspective, establishing the optimal recourse achievable while maintaining a given approximation is a fundamental combinatorial question about how significantly updates to the data can change the set of all approximately optimal solutions. Problems for which low-recourse algorithms exist, therefore, must possess some smoothness in their set of approximate solutions. Consistency is also closely related, and in some settings equivalent, to the notion of low-recourse algorithms in the online algorithms literature, which has received significant attention in recent years for various problems \cite{gu2013power,gupta2014online,lkacki2015power,megow2012power,epstein2014robust,sanders2004online,gupta2014maintaining,bernstein2019online,gupta2017online,gupta2020fully,bhattacharya2023chasing}. Here, the typical restrictions of online algorithms are relaxed by allowing some of the past decisions of the algorithm to be changed. 

From a practical perspective, for many modern computational tasks, it is indeed possible to make changes to the solution maintained by the algorithm, rather than having decisions be totally irrevocable. For instance, load balancers must make decisions for which machine to place an item on the fly, however, data may later be reshuffled across machines, although this data shuffling is a costly operation. Furthermore, clustering is often employed as a preliminary step as part of a larger ML pipeline, such as for feature engineering and classification tasks; in such settings, changing the clustering can involve costly recomputations of the downstream pipeline. 

This notion of consistent clustering was first introduced by \cite{lattanzi2017consistent}, who demonstrated that in the \emph{incremental setting}, where points are inserted but not deleted, one can maintain a $O(1)$-approximate solution for commonly studied $k$-clustering objectives, including $k$-center, $k$-means, and $k$-median with $O((k^2 \polylog n) / n)$ amortized recourse,\footnote{For $k$-center, they showed that the classic ``doubling algorithm'' of \cite{charikar2004incremental} already achieved the tight $O((k \log n) / n)$ recourse bound.} and also proved a $\Omega((k \log n)/n)$ lower bound. Their results were later strengthened by \cite{fichtenberger2021consistent} to a $O((k \cdot \polylog n)/n)$ amortized recourse bound, which is tight up to $\polylog n$ factors in the incremental setting. Thus, for incremental clustering, it is possible to achieve a recourse that is significantly sublinear in the number of updates to the data. 

Observe that achieving subconstant recourse bounds (when $n \gg k$) relies heavily upon the fact that the input is incremental. Namely, in the fully dynamic setting, this is clearly no longer possible: simply by repeatedly inserting and deleting a highly significant point, any bounded-approximation algorithm will be forced to add and remove that point from its set of centers after every time step. Thus, for nearly all $k$-clustering objectives, $\Omega(1)$ is a lower bound for the \textit{amortized} recourse of any fully dynamic clustering algorithms. On the other hand, to date, there has been little success in obtaining better algorithms that come close to meeting this lower bound. In fact, for most $k$-clustering tasks, such as $k$-center, $k$-means, and $k$-median, the best-known algorithm is simply to fully recluster the data after every update, resulting in a trivial $O(k)$ amortized recourse. 

The aforementioned discrepancy between the upper and lower bounds for consistent fully dynamic clustering algorithms represents a strong gap in the literature. The lack of improved algorithms for the fully dynamic setting, despite the strong progress in the incremental setting, gives some indication that perhaps nothing better can be done than repeated reclustering of the data. This state of affairs motivates the following key question:


\vspace{0.5em}

\begin{mdframed}
\begin{quote}
 \begin{center}
  {\it  Are there consistent $k$-clustering algorithms with $o(k)$ recourse in the fully dynamic setting?}
 \end{center}
 \end{quote}
\end{mdframed}
\vspace{-0.5em}

Furthermore, the existing incremental consistent algorithms for $k$-clustering~\cite{charikar2004incremental,lattanzi2017consistent,fichtenberger2021consistent} all leverage the fact that the cost of the optimal solution can only increase over time.
Specifically, whenever the (approximate) cost of the solution increases by a constant factor, they recompute the solution from scratch.
Because of that, even though their amortized recourse can be subconstant, the \emph{worst-case} recourse can be as large as $\Omega(k)$.
This motivates the second question:

\vspace{0.5em}
\begin{mdframed}
\begin{quote}
 \begin{center}
  {\it  Are there consistent $k$-clustering algorithms with $o(k)$ \emph{worst-case} recourse?}
 \end{center}
 \end{quote}
\end{mdframed}
\vspace{-0.5em}

In this work, we provide an affirmative answer to both above questions for the classic $k$-centers objective \cite{hsu1979easy,gonzalez1985clustering,hochbaum1986unified,feder1988optimal,bern1997approximation}. Recall that, given a metric space $(\mathcal{X},d)$ and a set of points
$P\subseteq\mathcal{X}$, the $k$-center problem is to output
a set $C\subset\mathcal{X}$ of at most $k$ centers such that the
maximum distance of any point $p\in P$ to the nearest center $c\in C$
is minimized; in other words, the goal is to minimize $\max_{p\in P}d(p,C)$ where $d(p,C)=\min_{c\in C}d(p,c)$. In the offline setting, $k$-center admits several well-known greedy $2$-approximation algorithms \cite{gonzalez1985clustering,hochbaum1986unified}. Moreover, it is known to be NP-hard to approximate the objective to
within a factor of $(2-\epsilon)$ for any constant $\eps>0$~\cite{hsu1979easy}.
Even restricted to Euclidean space, it is still NP-hard to approximate beyond a factor of $1.822$ \cite{feder1988optimal,bern1997approximation}. Thus, the majority of the literature on dynamic $k$-center focuses on obtaining constant approximations.

Our main result is a fully dynamic algorithm that maintains a constant approximation with constant amortized recourse, thereby matching the aforementioned consistency and approximation lower bounds up to a constant. Specifically, we prove the following theorem:

\begin{restatable}{theorem}{main}
\label{thm:main}
There exists a fully dynamic deterministic algorithm that maintains a constant-approximate solution of the $k$-center problem and obtains worst-case recourse of at most $1$ for an insertion and $2$ for a deletion.
\end{restatable}

Our algorithm works by assigning each point a \emph{rank}, which is a non-negative integer with the following property. For any $k'$ and at any time, the $k'$ points of the highest rank form a constant approximate solution to the $k'$-center problem.
Hence, in a sense, it solves the $k$-center problem for all possible values of $k$ at the same time.

While maintaining such a rank function is a standard idea in the dynamic clustering literature, it turns out that in our setting it is quite complicated to maintain the rank function while achieving constant recourse. 
We end up solving this problem by maintaining \emph{two} rank functions, a \emph{geometrical} rank function whose properties are directly tied to the distances between points of $P$, and a \emph{smooth} rank function that aims to approximate the geometrical rank function while only changing slowly. Ultimately, the properties of the geometric rank function control the  approximation ratio of our algorithm, while the smooth rank function ensures that we achieve constant recourse. The interplay between the two rank functions is quite subtle and we capture it by a hierarchical forest structure that we call a \emph{leveled forest}. 

\paragraph{Roadmap}
The paper is structured as follows. In \cref{sec:overview} we informally overview our algorithm and explain why it achieves both constant approximation ratio and worst-case constant recourse. \cref{sec:formal} then contains a formal analysis of the algorithm.

\section{Technical Overview}
\label{sec:overview}

In this section, we describe the intuition behind our algorithm. First, in \cref{sec:o1} we describe a so-called basic algorithm, a natural hierarchical algorithm for the $k$-center problem that achieves a constant approximate ratio but can have an arbitrary recourse. In \cref{sec:o2} we describe how we can build a combinatorial structure that we call a leveled forest alongside a run of the basic algorithm. This combinatorial structure at some point of the algorithm captures the important aspects of what happened in the past. Lastly, in \cref{sec:o3} we explain how the leveled forest leads to a natural improvement over the basic algorithm that yields constant recourse. 
We use standard notation that is formally introduced later in \cref{sec:notation}. 

\subsection{Basic algorithm}
\label{sec:o1}

We start by describing a \emph{basic algorithm}, a simple online hierarchical algorithm for $k$-center that already achieves a constant approximation ratio but not constant recourse. This algorithm serves as the basis of our more complicated algorithm achieving constant recourse. The main idea behind this basic algorithm is standard in the context of dynamic algorithms for $k$-center: at every step, we assign every point $p$ in the current point set $P$ its \emph{rank} $\ksi(p)$, which is a non-negative integer between $0$ and $\lceil \log_2 \Delta \rceil$.
Here, $\Delta$ denotes the maximum distance between any two points in $P$ throughout the algorithm.
For the purpose of this technical overview, we assume that $\Delta$ is known upfront and that the minimum distance between two points in $P$ is at least $1$.
We want to maintain that the ranks induce maximal independent sets on all scales. In particular, for every $i$, we maintain that all points of rank at least $i$ have the following two properties:
\begin{enumerate}
    \item Separation: For every two points $p,q$ with $\ksi(p), \ksi(q) \ge r$, we have have $d(p,q) \ge 2^r$. 
    \item Maximality: For every point $p$ there exists $p'$ with $\ksi(p') > \ksi(p)$ and $d(p, p') < 2^{\ksi(p) + 1}$; alternatively $\ksi(p)$ is strictly larger than $\ksi(p')$ for all $p' \in P\setminus \{p\}$.
\end{enumerate}

\paragraph{Constant approximation}
Any rank function that satisfies both the separation and the maximality property has the following property: Any solution formed by $k$ points with the largest rank is $O(1)$-approximate. 
To see this, first let $r$ be the rank of $k+1$-th point with the largest rank. On one hand, notice that the cost of the optimal solution satisfies $OPT \ge 2^{r} / 2$ since the separation property implies that if we imagine drawing balls of radius strictly less than $2^r/2$ around the $k+1$ points with largest ranks, these balls are disjoint. Hence, any solution with $k$ centers has cost at least $2^r/2$. 

On the other hand, starting at any point $p$ with rank $\ksi(p) \le r$, we can construct a chain of points starting at $p$ and ending at some point $q$ with rank $\ksi(q) > r$ as follows: The point $p$ has a neighbor $p_1$ with rank $\ksi(p_1) > \ksi(p) $ such that $d(p, p_1) \le 2^{\ksi(p) + 1}$ by the maximality property. The point $p_1$ has another neighbor $p_2$ with rank $\ksi(p_2) > \ksi(p_1) $ such that $d(p_1, p_2) \le 2^{\ksi(p_1)+1}$ etc. Continuing until we reach the first point $q$ with $\ksi(q) > r$ and summing up a geometric sequence, we conclude that the cost of covering $p$ is $ O(2^r) = O(OPT)$. 

\paragraph{Maintaining the rank function}
Next, we describe a natural way of updating the rank function after the insertion/deletion of a point. This way does not quite achieve the constant recourse but it is the basis for our later, more complicated, algorithm. 

\emph{Insertions:} Whenever a new point $p$ is inserted into $P$, we simply assign it the highest possible rank $\ksi(p)$ such that we maintain the separation property. Notice that doing this maintains the maximality property. 

\emph{Deletions:} After a point is deleted from our solution, we iterate over all other remaining points in the point set (in arbitrary order) and increase the rank of each point as much as possible to maintain the separation property. Doing this also maintains the maximality property.

The issue with our basic algorithm is deletions: Deletion of one point can imply a change in rank for arbitrarily many points in the point set! As our plan is to keep the top $k$ points with the highest rank as our solution, the naive basic algorithm can have an arbitrarily high recourse.

\subsection{Leveled forest}
\label{sec:o2}

Before moving on to a better algorithm, we next explain the crucial property of the basic algorithm. While we noticed that after the deletion of a point, an arbitrary number of other points can have their ranks increased, there is a structural observation to be made. 

Fix a rank $r$ and assume we delete a point $p$.
Let $\ksi$ and $\ksi'$ be the rank function before and after the deletion, respectively.
Consider any two points $q_1, q_2$ such that $\ksi(q_1) \le r, \ksi(q_2) \le r$, $\ksi'(q_1) > r$ and $\ksi'(q_2) > r$.
We observe that the new ranks of $q_1$ and $q_2$ cannot both increase by at least $2$, i.e., $\ksi'(q_1), \ksi'(q_2) \ge r+2$ is not possible. To see this, first note that we have $d(p, q_i) < 2^{\ksi(q_i) + 1}$ for $i = 1,2$ since $p$ must have been the point that prevented the ranks of $q_i$ from being increased before the deletion.
By triangle inequality, we thus have $d(q_1, q_2) < 2^{\ksi(q_1) + 1} + 2^{\ksi(q_2) + 1} \le 2^{r + 2}$. Hence if it were indeed true that $\ksi'(q_1), \ksi'(q_2) \ge r + 2$, the distance between $q_1$ and $q_2$ would contradict the separation property of the new rank function $\ksi'$. We can even strengthen this observation: out of all points that increased their rank to at least $r+1$ in the same step, at any later point in time, at most one of them can have a rank at least $r+2$.

A clean way to understand this observation is by drawing a diagram that we called a leveled forest (see \cref{fig:forest}). This is a forest such that its leaves correspond to the current points in $P$ and are at level $0$ (bottom nodes of each subpicture of \cref{fig:forest}). We additionally maintain a decomposition of the forest into edge-disjoint paths going up from the leaves such that the height of each path is equal to the rank of the respective point.

The forest is constructed during the execution of the algorithm as follows. Whenever a new point is inserted and gets rank $r$, we add a new leaf at level $0$ and then add a path of length $r$ from that leaf up. 

Deletion of a point is more complicated and explained in \cref{fig:forest}. 
Whenever a point is deleted, we first delete all edges from the corresponding upwards path from the forest; we also delete all vertices from the tree whose degree becomes zero (see the left and the middle picture of \cref{fig:forest} for the forest before and after this operation). 
Next, recall that the ranks of some points are increased to preserve the maximality property of the rank function. We split all the increase operations into batches where $r$-th batch corresponds to a group of points whose rank increases from $r$ to $r+1$. Starting at $r = 0$, we represent each $r$-th batch operation in the leveled forest by adding one additional node to level $r+1$ of the leveled forest and connect it with the rank $r$ nodes corresponding to the points in that batch (see the right picture in \cref{fig:forest}, thick edges correspond to four different points increasing their rank by one). To be more precise, we perform these updates of the leveled forest sequentially starting at $r = 0$; this is because one point can be in several batches if its rank increases by more than one.

We observe that this way of building the forest allows us to keep a natural decomposition of its edges into paths going upwards from the leaves, such that the length of each path equals the rank of the corresponding point. 

Importantly, our observation from this section implies that the picture from \cref{fig:forest} is not misleading. 
A priori, the way we build the leveled forest could lead to two edges going up from its inner node $w$ of rank $r+1$. But whenever this happens, it means that there was some point $p$ in the subtree of $w$ that increased its rank to at least $r+2$ and later, while $p$ was not yet deleted, some other point $q$ in the subtree of $w$ increased its rank to at least $r+2$. Both $p$ and $q$ raised their rank from $r$ to $r+1$ in the same step when the node $w$ was created. This is exactly the scenario forbidden by our observation.

\begin{figure}
    \centering
    \includegraphics[width = \textwidth]{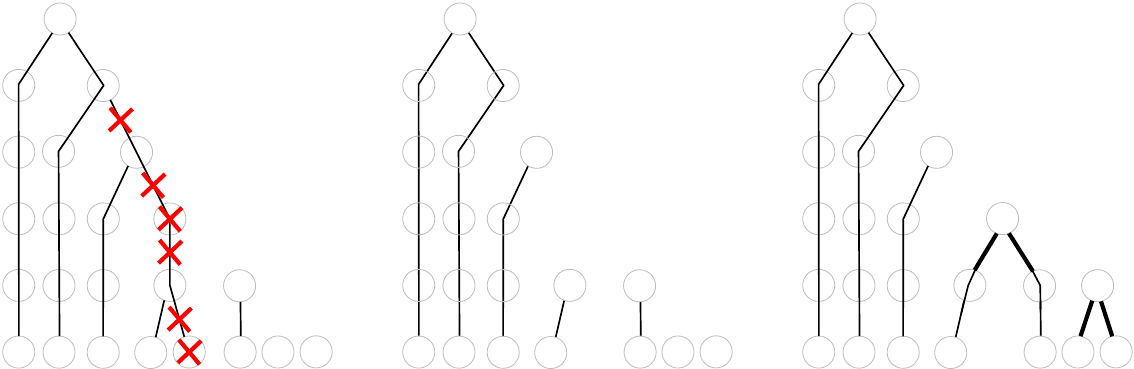}
    \caption{A leveled forest and its change after the deletion of one point. The leaves of the forest correspond to points in the current point set. Notice that the tree is decomposed into edge-disjoint paths that start at leaves and go only up. In the diagram, this can be seen in the left picture on the inner nodes of the tree with two children: only one path crosses the respective circle and continues up. \\
    Left: When a point is deleted, we first delete the whole path starting at the corresponding leaf from the forest. We also delete all childless inner nodes and the leaf itself (the two nodes with the red cross). \\
    Middle: The state of the forest after the deletion. Note that we maintain a decomposition of the forest into edge-disjoint paths. \\
    Right: After the deletion, some points increase their rank. For each rank, this corresponds to adding at most one node to the corresponding level. The newly created edges are thick; the decomposition to the edge-disjoint paths is updated accordingly. }
    \label{fig:forest}
\end{figure}

\begin{figure}
    \centering
    \includegraphics[width = \textwidth]{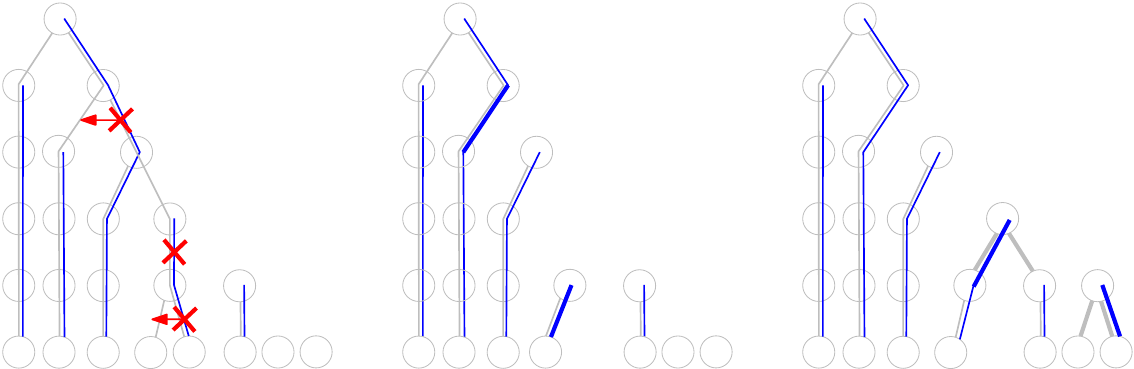}
    \caption{A leveled forest and its decomposition into vertex-disjoint paths (blue color), which go up from the leaves and define smooth ranks. \\
    Left: When a point is deleted, we have already seen in \cref{fig:forest} that the whole path corresponding to its geometrical rank is deleted. We need to deal with the blue edges on this path. If the top node of any such blue edge has another child, we make an edge to an arbitrary other child blue (red arrows). If there is no other child, we delete the blue edge. \\
    Middle: The state of the forest after the operation (the new blue edges are thick blue); note that the new blue paths still form a decomposition of the leveled forest into vertex-disjoint paths. \\
    Right: After the deletion, some points need to increase their geometric rank. Whenever a new node is added to the tree, we pick an edge to its arbitrary child and make it blue (thick blue edges). Note that after this operation we still have a valid decomposition of the forest into vertex-disjoint paths. 
    }
    \label{fig:forest2}
\end{figure}

\subsection{Smooth ranks and our algorithm}
\label{sec:o3}

From now on, we call the rank function from the basic algorithm the \emph{geometric rank}, $\ksi_G$. 
While using this rank function enables us to achieve the constant approximation ratio, its value can change abruptly for many points after every deletion.\footnote{We believe that with a certain potential function, one can prove that the basic algorithm achieves constant \emph{amortized} recourse. }

As a remedy, in our constant worst-case recourse algorithm, we maintain one more rank function, the \emph{smooth rank} $\ksi_S$, which is supposed to approximate the geometric rank, while only changing a bit in every step. 
Our final algorithm is then very similar to the basic algorithm, except for additionally maintaining the smooth ranks and outputting the $k$ points with the highest smooth rank as the output at every step. 

\paragraph{Defining smooth ranks}
How should the smooth rank be defined? The leveled forest leads to a straightforward definition. Recall that we keep a natural decomposition of the forest into \emph{edge}-disjoint paths that correspond to geometric ranks. We additionally maintain a separate decomposition of the forest into \emph{vertex}-disjoint paths going up from the leaves of the leveled forest. The smooth rank of every point is then defined as the length of the respective path (see \cref{fig:forest2}). Notice that we are not trying to keep any additional property tying the geometric and the smooth rank (e.g. requiring that the smooth rank of every point is at most its geometric rank etc.); the two rank functions are only tied together by the structure of the leveled forest. 

We now explain how we can maintain a decomposition into vertex-disjoint paths under insertions and deletions in a natural way. First, whenever a new point is inserted, we simply define its smooth rank to be equal to the geometric rank, i.e., the whole path corresponding to its geometric rank is colored blue. 

Next, consider the case when a point is deleted together with its path in the leveled forest corresponding to the geometric rank. See \cref{fig:forest2} that contains the same situation as \cref{fig:forest} from the perspective of smooth ranks. 

We iterate over edges deleted from the leveled forest. Whenever such an edge is used in our decomposition into vertex-disjoint paths (i.e., it is blue in \cref{fig:forest2}), we check if the higher node of that edge has any other child and if so, we add one such edge to our decomposition. Whenever the node does not have any other child, the node is anyway deleted from the leveled forest. See the left and the middle picture in \cref{fig:forest2} to see the smooth ranks before and after such a change. 

Recall that after deleting a point, some other points increase their geometric rank which corresponds to adding new nodes to the leveled forest. 
Whenever a new internal node is added, we simply pick an arbitrary  edge incident to it and add it to the set of vertex-disjoint paths (see \cref{fig:forest2}). This finishes the description of how we maintain a decomposition of the leveled forest into vertex-disjoint paths. This decomposition in turn defines smooth ranks and thus the description of our algorithm is finished. 

\paragraph{Constant recourse}
To argue that our algorithm has constant recourse, we notice that although in one step, many (in fact, $O(\log \Delta)$) points can change their smooth rank, there are only $O(1)$ changes ``per rank''. That is, for any $r$ there are only $O(1)$ points such that their smooth rank was at most $r$ before the update and more than $r$ after the update. Similarly, there are at most $O(1)$ points whose smooth rank decreases from more than $r$ to at most $r$. This implies that the set of $k$ points with the highest smooth rank changes by $O(1)$ in every step and we thus get constant recourse (a more detailed analysis shows that the recourse is in fact equal to $2$). 

\paragraph{Constant approximation}
Consider any point $p$ and assume that the geometric rank of the point with $k+1$-th largest geometric rank is equal to $r$. Recall from our discussion of the basic algorithm in \cref{sec:o1} that $OPT = \Omega(2^r)$ and that we can find a point $p'$ (possibly $p' = p$) such that $\ksi_G(p') \ge r+1$ and $d(p, p') = O(2^r) = O(OPT)$. 

Next, we show how to find a point $q$ such that $d(p', q) = O(OPT)$ and $q$ is included in our solution; this implies our algorithm has a constant approximation ratio. 

Locate the leaf corresponding to $p'$ in the leveled forest and walk up from that leaf $r + 1$ steps until reaching an internal node $w$ at level $r+1$: The fact that $\ksi_G(p') \ge r+1$ implies that such a node exists. This internal node is a vertex of one of the vertex-disjoint paths covering the forest; let $q$ be the point corresponding to this path (i.e., $q$ is the point corresponding to the bottommost vertex of the path). 

On one hand, the point $q$ is included as a center in our solution: This is because in any leveled forest, if there are at most $k$ edges reaching level $r+1$, then there are at most $k$ vertex-disjoint paths reaching that level. On the other hand, the distance between $p'$ and $q$ is $O(2^r)$: In fact, by the construction of the leveled forest, every pair of points that are in the subtree of $w$ in the leveled forest are of pairwise distance $O(2^r)$.

\section{Formal proof}
\label{sec:formal}

This section is devoted to the formal proof of \cref{thm:main}. Below is the technical statement proven in the rest of this section that implies \cref{thm:main}. 

\begin{restatable}{theorem}{maintechnical}
    Let $(\mathcal{X},d)$ be a metric space and $k,\Delta,n \in \mathbb{N}$. Let $P_1,P_2, \ldots,P_n$ be a sequence of point sets such that for every $i \in \{1,2,\ldots,n\}$, it holds that $P_i \subseteq \mathcal{X}$ and the pairwise distance between any two points in $P_i$ is at least $1$ and at most $\Delta$. Moreover, for every $i \in \{1,2,\ldots,n-1\}$, it holds that $|P_i \triangle P_{i+1}| = 1$. There exists a sequence of centers $C_1,C_2,\ldots, C_n$ such that for every $i \in \{1,2,\ldots,n\}$, it holds that $C_i \subseteq P_i$, $|C_i| \leq k$, $C_i$ is fully determined by $P_1,P_2,\ldots,P_i$ (together with ($\mathcal{X},d)$, $k$ and $\Delta$) and $C_i$ satisfies

    \[\max_{p \in P_i} d(p,C_i) \leq 24 \cdot \min_{C \subseteq \mathcal{X},|C| \leq k} \max_{p \in P_i} d(p,C).\]

    Moreover, for every $i \in \{1,2,\ldots,n-1\}$, if $P_{i+1} \setminus P_i \neq \emptyset$, then $|C_i \triangle C_{i+1}| \leq 2$ and if $P_i \setminus P_{i+1} \neq \emptyset$, then $|C_i \triangle C_{i+1}| \leq 4$.
\end{restatable}

Let us make a few remarks about the theorem. The time complexity of our algorithm in every step is dominated by the need of solving nearest-neighbour problems to determine whether a rank of a point can be increased. Thus, we spend a large polynomial time every step. It would be interesting to have an algorithm that needs only $\poly(k \log (n\Delta))$ time in every step. 

Second, we don't know whether the recourse of two after deletion is optimal or can be improved to one. It would be interesting to see an example showing that the recourse cannot be improved or an algorithm with better recourse.

The rest of the section is structured as follows. 
First, we review some basic notation in \cref{sec:notation}. Next, \cref{sec:invariants} is devoted to the formalization of the rank functions, leveled forest, and its decomposition to edge- and vertex-disjoint paths that correspond to ranks. The meat of this section is \cref{sec:algorithm} where we describe our algorithm for updating the leveled forest in case of an Insert and Delete operation, thus formalizing the algorithm sketched in \cref{fig:forest,fig:forest2}. We finish the proof in \cref{sec:finish}.

\subsection{Basic notation}
\label{sec:notation}
We use $A \triangle B$ to denote the symmetrical difference of two sets $A$ and $B$. 
We use $d(p,q)$ to denote the distance between $p$ and $q$ in the input metric space $\mathcal{X}$. We use $d(p,C) = \min_{c \in C} d(p, c)$. 
For a pointset $P$ and its subset $C$, we define

\[COST(P,C) = \max_{p \in P} d(p,C).\]

For every $k \in \mathbb{N}$, we define

\[OPT_k(P) := \min_{C \subseteq \mathcal{X}, |C| \leq k} COST(P,C).\]

\subsection{Basic definitions and invariants}\label{sec:invariants}

This section lists the invariants that we keep throughout our algorithm. Mostly, we formalize the notions of rank function and the leveled forest that, as we discussed in \cref{sec:overview}, are crucial for our algorithm. 

\paragraph{Rank functions}
We next define the rank function and formalize its properties and implications. 

\begin{definition}[Rank Function $\ksi$ and Ordered Rank $\ksi^*$]
Let $(\mathcal{X},d)$ be a metric space.
A rank function for a point set $P \subseteq \mathcal{X}$ assigns each $p \in P$ a nonnegative integer, denoted by $\ksi(p)$. Additionally, for a given $i \in \{1, 2, \ldots, |P|\}$, let $\ksi^*(i)$ denote the $i$-th largest rank assigned by $\ksi$.
\end{definition}

\begin{definition}[Separation Property]
\label{def:separation}
    Let $(\mathcal{X},d)$ be a metric space, $P \subseteq \mathcal{X}$ and $\ksi \colon P \mapsto \mathbb{N}_0$ be a rank function. We say that $\ksi$ satisfies the separation property iff for any two distinct points $p_1,p_2 \in P$, it holds that $d(p_1,p_2) \geq 2^{\min(\ksi(p_1),\ksi(p_2))}$.
\end{definition}
\begin{lemma}[Lower Bound on OPT based on Separation]
\label{lem:separation}
 Let $(\mathcal{X},d)$ be a metric space, $P \subseteq \mathcal{X}$ and $\ksi \colon P \mapsto \mathbb{N}_0$ be a rank function satisfying the separation property. Then, for every $k \in \{1,2,\ldots,|P|-1\}$, it holds that 

\[OPT_k(P) \geq 0.5 \cdot 2^{\ksi^*(k+1)}.\]
\end{lemma}
\begin{proof}
We show that for every set $C \subseteq \mathcal{X}$ with $|C| \leq k$, there exists a point $p \in P$ such that $d(p,C) \geq 0.5 \cdot 2^{\ksi^*(k+1)}$. This then implies that $OPT_k(P) \geq 0.5 \cdot 2^{\ksi^*(k+1)}$.

By definition of $\ksi^*(k+1)$, there exists a set $P' \subseteq P$ of size $k+1$ such that for every $p \in P'$, $\ksi(p) \geq \ksi^*(k+1)$. As $\ksi$ satisfies the separation property, the pairwise distance between any two points in $P'$ is at least $2^{\ksi^*(k+1)}$.

Consider any set $C \subseteq \mathcal{X}$ with $|C| \leq k$. By the pigeonhole principle, there must exist two distinct points $p_1, p_2 \in P'$ and a center $c \in C$ such that $d(p_1,c) = d(p_1,C)$ and $d(p_2,c) = d(p_2,C)$. As $d(p_1, p_2) \geq 2^{\ksi^*(k+1)}$, the triangle inequality inplies that $d(p_1,c) \geq 0.5 \cdot 2^{\ksi^*(k+1)}$ or $d(p_2,c) \geq 0.5 \cdot 2^{\ksi^*(k+1)}$ must hold. Thus, we have found a point $p \in \{p_1,p_2\}$ such that $d(p,C) \geq 0.5 \cdot 2^{\ksi^*(k+1)}$.

Since the set $C$ was chosen arbitrarily, we can conclude that $OPT_k(P) \geq 0.5 \cdot 2^{\ksi^*(k+1)}$.
\end{proof}

\begin{definition}[Maximality Property]
\label{def:maximality}
    Let $(\mathcal{X},d)$ be a metric space, $P \subseteq \mathcal{X}$ and $\ksi \colon P \mapsto \mathbb{N}_0$ be a rank function. We refer to a point $p \in P$ as maximal with respect to $\ksi$ if either $\ksi(p) > \ksi(p')$ for every $p' \in P \setminus \{p\}$ or if there exists $p' \in P$ with $\ksi(p') > \ksi(p)$ and $d(p,p') < 2^{\ksi(p) + 1}$. We say that $\ksi$ is maximal if any point in $P$ is maximal with respect to $\ksi$.
\end{definition}

\begin{lemma}
\label{lem:maximality}
  Let $(\mathcal{X},d)$ be a metric space, $P \subseteq \mathcal{X}$ and $\ksi \colon P \mapsto \mathbb{N}_0$ be a rank function satisfying the maximality property. For every $p \in P$ and $i \in \{2,3,\ldots,|P|\}$, there exists a point $q \in P$ with $\ksi(q) > \ksi^*(i)$ and $d(p,q) \leq 4 \cdot 2^{\ksi^*(i)}$.

\end{lemma}
\begin{proof}
We prove by induction on $j \in \{0,1,\ldots,\ksi^*(i)\}$ that for every point $p$ with $\ksi(p) \geq \ksi^*(i) - j$, there exists a point $q \in P$ with $\ksi(q) > \ksi^*(i)$ and $d(p,q) \leq 2 \cdot \sum_{\ell = \ksi^*(i)-j}^{\ksi^*(i)} 2^{\ell} \leq 4 \cdot 2^{\ksi^*(i)}$.

\textbf{Base case:} For $j=0$, let $p$ be a point with $\ksi(p) = \ksi^*(i)$ (if $\ksi(p) > \ksi^*(i)$, we can set $q = p$). As $\ksi$ is maximal and $p$ is not the point with the largest rank, there exists a point $q \in P$ such that $\ksi(q) > \ksi(p) = \ksi^*(i)$ and $d(p, q) < 2^{\ksi(p) + 1} = 2^{\ksi^*(i)+1} = 2 \cdot \sum_{\ell = \ksi^*(i)}^{\ksi^*(i)} 2^{\ell} \leq 4 \cdot 2^{\ksi^*(i)}$. Therefore, the base case holds.

\textbf{Inductive step:} Assume the statement holds for some $j \geq 0$. We show it also holds for $j+1$. Let $p$ be a point with $\ksi(p) = \ksi^*(i) - (j+1)$. As $\ksi$ is maximal, there exists a point $q' \in P$ such that $\ksi(q') > \ksi(p)$ and $d(p, q') < 2^{\ksi(p) + 1} = 2^{\ksi^*(i) - j}$. By the inductive hypothesis, there exists a point $q \in P$ with $\ksi(q) > \ksi^*(i)$ and $d(q', q) \leq 2 \cdot \sum_{\ell = \ksi^*(i) - j}^{\ksi^*(i)} 2^{\ell}$. By the triangle inequality, $d(p, q) \leq d(p, q') + d(q', q) \leq 2 \cdot \sum_{\ell = \ksi^*(i) - (j+1)}^{\ksi^*(i)} 2^{\ell} \leq 4 \cdot 2^{\ksi^*(i)}$. This completes the inductive step and hence the proof.
\end{proof}

\paragraph{Valid tuples}
The following \cref{def:valid_tuple} of a valid tuple  and \cref{lem:sufficient_for_good_approximation} formalize the properties that two rank functions $\ksi_G$ (a geometric rank) and $\ksi_S$ (smooth rank) need to satisfy so as to imply constant approximate ratio. 

\begin{definition}[Valid Tuple]
\label{def:valid_tuple}
 Let $(\mathcal{X},d)$ be a metric space, $P \subseteq \mathcal{X}$ and $\ksi_G,\ksi_S \colon P \mapsto \mathbb{N}_0$ be two rank functions for $P$.
 We refer to $(\ksi_G,\ksi_S)$ as a valid tuple for $P$ iff the following properties are satisfied:
\begin{enumerate}
    \item The rank function $\ksi_G$ satisfies both the separation and the maximality property.
    \item It holds that $\ksi^*_G(i) \geq \ksi^*_S(i)$ for every $i \in \{1,2,\ldots,|P|\}$.
    \item For every $r \in \mathbb{N}_0$ and $p \in P$ with $\ksi_G(p) \geq r$, there exists a point $q \in P$ with $\ksi_S(q) \geq r$ and $d(p,q) \leq 4 \cdot 2^r$.
\end{enumerate}
\end{definition}

\begin{lemma}
\label{lem:sufficient_for_good_approximation}
 Let $(\mathcal{X},d)$ be a metric space, $P \subseteq \mathcal{X}$ and $(\ksi_G,\ksi_S)$ be a valid tuple for $P$.
Let $k \in \{1,2,\ldots,|P|-1\}$ and $C \subseteq P$ such that each point $p \in P$ with $\ksi_S(p) > \ksi_S^*(k+1)$ is contained in $C$. Then,

  \[COST(P,C) \leq 24 \cdot OPT_k(P).\] 
\end{lemma}
\begin{proof}
 As $\ksi_G$ satisfies the separation property, we can conclude that $OPT_k(P) \geq 0.5 \cdot 2^{\ksi_G^*(k+1)}$ using \cref{lem:separation}. It thus suffices to show that $COST(P,C) \leq 12 \cdot 2^{\ksi_G^*(k+1)}$.
 Consider an arbitrary point $p \in P$. As $\ksi_G$ satisfies the maximality property, \cref{lem:maximality} implies that there exists a point $q \in P$ with $\ksi_G(q) > \ksi_G^*(k+1)$ and $d(p,q)\leq 4 \cdot 2^{\ksi_G^*(k+1)}$. Using the third property of a valid tuple, there exists a point $q' \in P$ with $\ksi_S(q') > \ksi^*_G(k+1)$ and $d(q,q') \leq 4 \cdot 2^{\ksi^*_G(k+1) + 1}$. As $\ksi_S(q') > \ksi^*_G(k+1) \geq \ksi^*_S(k+1)$, we get that $q' \in C$ by the second property. Using triangle inequality, we get

 \[d(p,C) \leq d(p,q') \leq d(p,q) + d(q,q') \leq 4 \cdot 2^{\ksi_G^*(k+1)} + 4 \cdot 2^{\ksi^*_G(k+1) + 1} = 12 \cdot 2^{\ksi_G^*(k+1)}\]
 which concludes the proof.
\end{proof}

\paragraph{Definitions related to leveled forest}
We next describe the formal definition of the leveled forest. Formally, a leveled forest is a structure that enables us to maintain a valid tuple from \cref{def:valid_tuple}. 

\begin{definition}[Leveled Forest]
\label{def:leveled_forest}
Let $(\mathcal{X},d)$ be a metric space and $P \subseteq \mathcal{X}$. We refer to $F$ as a leveled forest for $P$ if the following conditions are satisfied:

\begin{enumerate}
    \item $F$ is a rooted forest with edges oriented towards the respective root.
    \item The set of leafs in $F$ is equal to $P$.
    \item For each rooted tree $T$ in $F$, every leaf in $T$ has the same distance to the root of $T$.
\end{enumerate}
For consistency, if $P = \emptyset$, we refer to $F$ as a leveled forest for $P$ if both the vertex and the edge set is empty.
\end{definition}

\begin{definition}[Height Function $h_F$]
\label{def:height_function}
    Let $(\mathcal{X},d)$ be a metric space, $P \subseteq \mathcal{X}$ and $F$ be a leveled forest for $P$.
 We associate a height function $h_F \colon V(F) \mapsto \mathbb{N}_0$ with $F$, where for each node $v$ in $F$, $h_F(v)$ represents the number of edges on the path from $v$ to each leaf in its subtree.
\end{definition}

The following definition of a valid triple formalizes that we want to maintain not just the forest itself, but also its two decompositions in edge-disjoint and vertex-disjoint paths that correspond to the geometric and smooth ranks respectively. 

\begin{definition}[Valid Triple]
\label{def:valid_triple}
Let $(\mathcal{X},d)$ be a metric space, $P \subseteq \mathcal{X}$, $F$ be a leveled forest for $P$ and $\ksi_G, \ksi_S$ be two rank functions for $P$. We say that $(F, \ksi_G, \ksi_S)$ is a valid triple for $P$ if it satisfies the following conditions:
\begin{enumerate}
    \item The rank function $\ksi_G$ satisfies the separation property.
    \item For each point $p$ in $P$, the values of $\ksi_G(p)$ and $\ksi_S(p)$ are at most the length of the path from $p$ to its root in $F$.
    \item For each node $v$ in the forest $F$, there is exactly one leaf $p_v$ in its subtree such that $\ksi_S(p_v) \geq h_F(v)$.
    \item For each edge $(u,v)$ in $F$, where $v$ is the parent of $u$, there is exactly one leaf $p_{uv}$ in the subtree of $u$ such that $\ksi_G(p_{uv}) \geq h_F(v)$.
    \item Let $v$ be a node in $F$, and $p_1$ and $p_2$ be two distinct leaves in the subtree of $v$ with $\min(\ksi_G(p_1), \ksi_G(p_2)) \geq h_F(v)$. Then, $d(p_1, p_2) < 2^{h_F(v) + 1}$.
\end{enumerate}
\end{definition}

\paragraph{Valid triples imply valid tuples}
We show that maintaining a leveled forest in the sense of \cref{def:valid_triple} implies that the corresponding rank function have the required properties in the sense of \cref{def:valid_tuple}. First, we prove that for any node $v$ in the leveled forest, all points in its subtree are close to each other.

\begin{lemma}
\label{lem:valid_triple_implies_small_distance}
    Let $(\mathcal{X},d)$ be a metric space, $P \subseteq \mathcal{X}$ and $(F,\ksi_G,\ksi_S)$ be a valid triple for $P$.
    Consider a common ancestor $v$ of two leaves $q$ and $p$ in $F$, where $\ksi_G(q) \geq h_F(v)$. Then, the distance $d(q,p)$ between $q$ and $p$ is at most $\sum_{i=1}^{h_F(v)} 2^{i+1} = 4 \cdot 2^{h_F(v)} - 4$.
\end{lemma}
\begin{proof}
We prove the statement by induction on $h_F(v)$.

\textbf{Base case:} Assume $h_F(v) = 1$ and let $p$ and $q$ be two distinct leaves in the subtree of $v$. As $(F,\ksi_G,\ksi_S)$ is a valid triple, the fourth property implies that both $\ksi_G(p) \geq 1$ and $\ksi_G(q) \geq 1$. Thus, the fifth property of a valid triple implies that $d(q,p) < 2^{h_F(v) + 1} = \sum_{i=1}^{h_F(v)} 2^{i+1}$. Therefore, the base case holds.

\textbf{Inductive step:} Assume the statement holds for some $j$. We show that it also holds for $j+1$. Consider a common ancestor $v$ of two leaves $p$ and $q$ with $h_F(v) = j+1$ and $\ksi_G(q) \geq h_F(v)$.
Let $u_p$ and $u_q$ be the two children of $v$ such that $p$ is in the subtree of $u_p$ and $q$ is in the subtree of $u_q$. If $u_p = u_q$, which occurs when $v$ is not the lowest common ancestor of $p$ and $q$, we can directly use induction to conclude that $d(q,p) \leq \sum_{i=1}^{h_F(u_p)} 2^{i+1} \leq \sum_{i=1}^{h_F(v)} 2^{i+1}$.
Thus, it remains to consider the case where $u_p \neq u_q$. Let $p_{u_pv}$ be the leaf in the subtree of $u_p$ satisfying $\ksi_G(p_{u_pv}) \geq h_F(v)$. Such a node $p_{u_pv}$ is guaranteed to exist because of the fourth property of a valid triple. Using induction, we have $d(p_{u_pv},p) \leq \sum_{i=1}^{h_F(u_p)} 2^{i+1}$. By the fifth property, we have $d(q,p_{u_pv}) < 2^{h_F(v) + 1}$. Therefore, using the triangle inequality, we can conclude that:
\[d(q,p) \leq d(q,p_{u_pv}) + d(p_{u_pv},p) \leq 2^{h_F(v) + 1} + \sum_{i=1}^{h_F(u_p)} 2^{i+1} = \sum_{i=1}^{h_F(v)} 2^{i+1}.\]
This completes the inductive step, and hence the proof.
\end{proof}

Next, we prove that valid triples imply valid tuples. 
In view of \cref{lem:sufficient_for_good_approximation}, this shows that maintaining the leveled forest in the sense of \cref{def:valid_triple} is enough to achieve constant approximation.

\begin{lemma}[Valid Triple together with Maximality implies Valid Tuple]
\label{lem:valid_triple_maximality_implies_valid_tuple}
   Let $(\mathcal{X},d)$ be a metric space, $P \subseteq \mathcal{X}$ and $(F,\ksi_G,\ksi_S)$ be a valid triple for $P$. Assume that $\ksi_G$ additionally satisfies the maximality property (\cref{def:maximality}). Then, $(\ksi_G,\ksi_S)$ is a valid tuple for $P$.
\end{lemma}
\begin{proof}
By definition, $\ksi_G$ satisfies both the separation and the maximality property.

To show that the second property of a valid tuple holds, consider an arbitrary $i \in \{1,2,\ldots,|P|\}$. We have to show that $\ksi^*_G(i) \geq \ksi^*_S(i)$. First, note that is suffices to give an injective mapping that maps each point $p \in P$ with $\ksi_S(p) \geq \ksi^*_S(i)$ to a point $p' \in P$ with $\ksi_G(p') \geq \ksi^*_S(i)$. Now, consider some $p \in P$ with $\ksi_S(p) \geq \ksi^*_S(i)$. Let $v$ be the ancestor of $p$ in $F$ with $h_F(v) = \ksi^*_S(i)$. The node $v$ is well-defined because of the second property of a valid triple. 
We now map $p$ to an arbitrary descendant $p'$ of $v$ with $\ksi_G(p') \geq h_F(v) = \ksi^*_S(i)$. Such a descendant is guaranteed to exist because of the fourth property. The third property, or more concretely that $p$ is the only descendant of $v$ with $\ksi_S(p) \geq h_F(v)$, implies that the mapping is indeed injective. 

To show the third property, consider an arbitrary $r \in \mathbb{N}_0$ and $p \in P$ with $\ksi_G(p) \geq r$. We have to show that there exists a point $q \in P$ with $\ksi_S(q) \geq r$ and $d(p,q) \leq 4 \cdot 2^r$.
Let $v$ be the ancestor of $p$ with $h_F(v) = r$. Let $q$ be the leaf in the subtree of $v$ with $\ksi_S(q) \geq r$, guaranteed by the third property. As $v$ is a common ancestor of $p$ and $q$ and $\ksi_G(p) \geq h_F(v)$, we can use \cref{lem:valid_triple_implies_small_distance} to conclude that $d(p,q) \leq 4 \cdot 2^{h_F(v)} - 4 \leq 4 \cdot 2^r$, which finishes the proof.
\end{proof}

\subsection{Algorithm}\label{sec:algorithm}

This section is devoted to describing our algorithm we run whenever a new point is inserted (\cref{alg:insert}) and whenever a point of the current solution is deleted (\cref{alg:delete}). These algorithms formalize the algorithm sketched in \cref{fig:forest,fig:forest2}, making precise what we do in certain special cases. 

\subsubsection{Insertion}

In this section we prove \cref{lem:insertion} that shows that we can maintain a valid triple (i.e., a leveled forest and its decomposition) after a point is inserted. 

\begin{algorithm}
\caption{Insert Operation}
\label{alg:insert}
\textbf{Input:} $P \subseteq \mathcal{X}$ for some metric space $(\mathcal{X},d)$ and the minimum pairwise distance in $P$ is at least $1$\\
$(F,\ksi_G,\ksi_S)$ is valid triple for $P$ and $\ksi_G$ additionally satisfies the maximality property \\
$q \in \mathcal{X}$ is an arbitrary point satisfying $d(q,p) \geq 1$ for every $p \in P$\\
$\Delta \in \mathbb{N}$ is an upper bound on the maximum pairwise distance between any two points in $P \cup \{q\}$ \\
If $P \neq \emptyset$, then there exists $p \in P$ with $\ksi_G(p) \geq \lceil \log(\Delta) + 1\rceil$ \\
\textbf{Output:} $(F', \ksi'_G, \ksi'_S)$ is a valid triple for $P \cup \{q\}$, $\ksi'_G$ satisfies the maximality property, there exists $p \in P \cup \{q\}$ with $\ksi'_G(p) \geq \lceil\log(\Delta) + 1\rceil$ and $\ksi'_S(p) = \ksi_S(p)$ for every $p \in P$.

\begin{algorithmic}[1]
\Procedure{Insert}{$q$,$P$,$F$,$\ksi_G,\ksi_S$,$\Delta$} 
    \State $i_q \leftarrow \max \{i \in \{0,1,\ldots, \lceil \log(\Delta) + 1 \rceil\}\mid d(q,p) \geq 2^{\min(i,\ksi_G(p))} \text{ for every $p \in P$}\}$
    \State Obtain $F'$ from $F$ by first adding $q$ as a new leaf, and then attaching a path of length $i_q$ to $q$ with edges oriented away from $q$.
    \State $\ksi'_G(q) = \ksi'_S(q) = i_q$ 
    \State $\ksi'_G(p) = \ksi_G(p)$ and $\ksi'_S(p) = \ksi_S(p)$ for every $p \in P$ 
    \State \Return $(F', \ksi'_G, \ksi'_S)$
\EndProcedure
\end{algorithmic}
\end{algorithm}

\begin{lemma}[Insertion Lemma]
    \label{lem:insertion}
    Let $(\mathcal{X},d)$ be a metric space and $P \subseteq \mathcal{X}$ be a point set with minimum pairwise distance at least $1$. 
    Let $(F,\ksi_G,\ksi_S)$ be a valid triple for point set $P$ and $\ksi_G$ additionally satisfies the maximality property. Let $q \in \mathcal{X}$ with $d(q,p) \geq 1$ for every $p \in P$. Let $\Delta \in \mathbb{N}$ be an upper bound on the maximum pairwise distance between any two points in $P \cup \{q\}$. If $P \neq \emptyset$, then we additionally assume that there exists a point $p \in P$ with $\ksi_G(p) = \lceil \log(\Delta) + 1 \rceil$. 
    Let $(F',\ksi'_G,\ksi'_S) \leftarrow \texttt{INSERT}(q,P,F,\ksi_G,\ksi_S,\Delta)$ (\cref{alg:insert}). Then, $(F',\ksi'_G,\ksi'_S)$ is a valid triple for $P \cup \{q\}$, $\ksi'_G$ satisfies the maximality property, there exists $p \in P \cup \{q\}$ with $\ksi'_G(p) \geq \lceil \log(\Delta) + 1 \rceil$ and $\ksi'_S(p) = \ksi_S(p)$ for every $p \in P$.
\end{lemma}
\begin{proof}
First, note that we assume that $d(q,p) \geq 1 = 2^{\min(0,\ksi_G(p))}$ for every $p \in P$.
Therefore, $i_q$ is well-defined. We first check that $\ksi_G'$ satisfies the separation property.
To check the separation property, consider two distinct points $p_1,p_2 \in P \cup \{q\}$.
We have to show that $d(p_1,p_2) \geq 2^{\min(\ksi'_G(p_1),\ksi'_G(p_2))}$.
If both $p_1$ and $p_2$ are in $P$, then this follows from $\ksi'_G(p) = \ksi_G(p)$ for every $p \in P$ and our assumption that $\ksi_G$ satisfies the separation property.
If either $p_1 = q$ or $p_2 = q$, then this directly follows from how $i_q$ is defined and that $\ksi'_G(q) = i_q$. 

We next verify that $\ksi'_G$ satisfies the maximality property. Consider an arbitrary $p \in P \cup \{q\}$.
We have to show that either $\ksi'_G(p) > \ksi'_G(p')$ for every $p' \in (P \cup \{q\}) \setminus \{p\}$ or there exists some $p' \in P \cup \{q\}$ with $\ksi'_G(p') > \ksi'_G(p)$ and $d(p,p') < 2^{\ksi'_G(p) + 1}$. 
First, if $P = \emptyset$, then $\ksi'_G$ clearly satisfies the maximality property.
It thus remains to consider the case $P \neq \emptyset$.
We first check that the newly inserted point $q$ satisfies the maximality property. First, note that $i_q < \lceil \log(\Delta) + 1 \rceil$.
This follows from our assumption that there exists a point $p \in P$ with $\ksi_G(p) \geq \lceil \log(\Delta) + 1\rceil$ and therefore $d(q,p) \leq \Delta < 2^{\min(\lceil \log(\Delta) + 1\rceil,\ksi_G(p))}$.
Thus, from the way $i_q$ is defined, it follows that there exists a point $p \in P$ with $d(q,p) < 2^{\min(i_q + 1,\ksi_G(p))} \leq 2^{\ksi'_G(q) + 1}$ and $d(q,p) \geq 2^{\min(i_q,\ksi_G(p))}$.
In particular, $\ksi_G(p) > i_q$, which directly implies $\ksi'_G(p) > \ksi'_G(q)$, and therefore the maximality property holds for $q$.
Next, we check that the maximality property also holds for an arbitrary $p \in P$.
First, consider the case that $\ksi_G(p) > \ksi_G(p')$ for every $p' \in P \setminus \{p\}$.
Then, by our assumption we get that $\ksi_G(p) \geq \lceil \log(\Delta) + 1 \rceil$.
We have argued above that $\ksi'_G(q) < \lceil \log(\Delta) + 1 \rceil$ and therefore $\ksi'_G(p) > \ksi'_G(p')$ for every $p' \in (P \cup \{q\}) \setminus \{p\}$.
Thus, it remains to check the case when there exists $p' \in P \setminus \{p\}$ with $\ksi_G(p) \leq \ksi_G(p')$.
As $\ksi_G$ satisfies the maximality property, it follows that there exists $p' \in P$ with $\ksi_G(p') > \ksi_G(p)$ and $d(p,p') < 2^{\ksi_G(p) + 1}$.
As $\ksi'_G(p) = \ksi_G(p)$ for every $p \in P$, we can thus conclude that $\ksi'_G$ indeed satisfies the maximality property.

We also have to verify that there exists $p \in P \cup \{q\}$ with $\ksi'_G(p) \geq \lceil \log(\Delta) + 1\rceil$.
If $P \neq \emptyset$, then this follows from our assumption that there exists $p \in P$ with $\ksi'_G(p) = \ksi_G(p) \geq \lceil \log(\Delta) + 1\rceil$
. If $P = \emptyset$, then it follows from the way $i_q$ is defined that $\ksi'_G(q) = \lceil \log(\Delta) + 1\rceil$. 
The fact that $\ksi'_S(p) = \ksi_S(p)$ for every $p \in P$ follows by definition.

It thus remains to verify that $(F',\ksi'_G,\ksi'_S)$ is a valid triple for $P \cup \{q\}$.
We have already verified above that $\ksi'_G$ satisfies the separation property.
All other properties straightforwardly follow from our assumption that $(F,\ksi_G,\ksi_S)$ is a valid triple for $P$ and the way $F',\ksi'_G$ and $\ksi'_S$ are derived from $F,\ksi_G$ and $\ksi_S$.
\end{proof}

\subsubsection{Deletion}

In this section, we show how we maintain the leveled forest during deletion. This section formalizes the algorithm sketched in \cref{fig:forest,fig:forest2}. 

\paragraph{RankDecrease operation}
When a point is deleted, we implement this deletion by repeatedly decrementing its rank by one. After each decrement, we want to keep a valid triple (i.e., leveled forest and its decomposition). This is done in \cref{alg:rankdecrease} and in \cref{lem:rankdecrease}. Intuitively, \cref{alg:rankdecrease} corresponds to one change in the left picture in \cref{fig:forest2}.

\begin{algorithm}
\caption{RankDecrease Operation}
\label{alg:rankdecrease}
\textbf{Input:} $P \subseteq \mathcal{X}$ for some metric space $(\mathcal{X},d)$\\
$(F,\ksi_G,\ksi_S)$ is a valid triple for $P$ \\
$q$ is a point in $P$ with $\ksi_G(q) \geq 1$ \\
\textbf{Output:} $(F', \ksi'_G, \ksi'_S)$ is a valid triple for point set $P$ with $\ksi'_G(q) = \ksi_G(q) - 1$.
\begin{algorithmic}[1]
\Procedure{RankDecrease}{$q$,$P$,$F$,$\ksi_G$,$\ksi_S$}
\State $\ksi'_G(p) =
\begin{cases}
    \ksi_G(p), & \text{if } p \in P  \setminus \{q\} \\
    \ksi_G(q) - 1, & \text{if } p = q
\end{cases}$
\State Let $u$ and $v$ be the ancestors of $q$ with $h_F(u) = \ksi_G(q) - 1$ and $h_F(v) = \ksi_G(q)$, respectively.
\State Let $p_v$ be the leaf in the subtree of $v$ with $\ksi_S(p_v) \geq h_F(v)$.
\If{$u$ is the only child of $v$ in $F$}
 \State We obtain $F'$ from $F$ by removing both $v$ and $(u,v)$.
  \State $\ksi'_S(p) =
\begin{cases}
    \ksi_S(p), & \text{if } p \in P  \setminus \{p_v\} \\
    h_F(u), & \text{if } p = p_v
\end{cases}$
 \Else
 \State Let $u'$ be an arbitrary child of $v$ not equal to $u$.
 \State We obtain $F'$ from $F$ by removing $(u,v)$.

\If{$p_v$ is in the subtree of $u$}

\State Let $p_{u'}$ be the leaf in the subtree of $u'$ with $\ksi_S(p_{u'}) \geq h_F(u')$.
 \State $\ksi'_S(p) =
\begin{cases}
    \ksi_S(p), & \text{if } p \in P  \setminus \{p_v,p_{u'}\} \\
    h_F(u), &  \text{if } p = p_v \\
    \ksi_S(p_v), & \text{if } p = p_{u'}
\end{cases}$
\Else
 \State $\ksi'_S(p) = \ksi_S(p)$ for every $p \in P$.
 \EndIf
\EndIf
\State \Return $(F', \ksi'_G, \ksi'_S)$
\EndProcedure
\end{algorithmic}
\end{algorithm}

\begin{lemma}[RankDecrease Lemma]
    \label{lem:rankdecrease}
    Let $(\mathcal{X},d)$ be a metric space, $P \subseteq \mathcal{X}$ be a point set and $(F,\ksi_G,\ksi_S)$ be a valid triple for $P$. Let $q$ be a point in $P$ with $\ksi_G(q) \geq 1$. Let $(F',\ksi'_G,\ksi'_S) \leftarrow \texttt{RANKDECREASE}(q,P,F,\ksi_G,\ksi_S)$ (\cref{alg:rankdecrease}). Then, $(F',\ksi'_G,\ksi'_S)$ is a valid triple for $P$. Also, the distance between $q$ and its root in $F'$ is $\ksi_G(q) - 1$.
Let $v$ denote the ancestor of $q$ with $h_F(v) = \ksi_G(q)$. 
There exists at most one point $p_{dec} \in P$ satisfying $\ksi'_S(p_{dec}) < \ksi_S(p_{dec})$. Moreover, $p_{dec}$ additionally satisfies $\ksi'_S(p_{dec}) = \ksi_G(q) - 1$ and if $v$ is a root, then also $\ksi_S(p_{dec}) = \ksi_G(q)$.
There also exists at most one point $p_{inc} \in P$ satisfying $\ksi'_S(p_{inc}) > \ksi_S(p_{inc})$. Moreover, $p_{inc}$ additionally satisfies $\ksi_S(p_{inc}) = \ksi_G(q) - 1$ and if $v$ is a root, then also $\ksi'_S(p_{inc}) = \ksi_G(q)$.
\end{lemma}
\begin{proof}
We first show that the algorithm is well-defined.
As $(F,\ksi_G,\ksi_S)$ is a valid triple for $P$ and $q \in P$, we know that $q$ is a leaf in $P$.
Moreover, the value of $\ksi_G(q)$ is at most the length of the path from $q$ to its root.
Therefore, $q$ has indeed exactly one ancestor in $F$ of distance $\ksi_G(q)$ and one ancestor in $F$ of distance $\ksi_G(q) - 1$.
Thus, it follows from the definition of the height function $h_F$ that $u$ and $v$ are well-defined and also that $u$ is a child of $v$. 
Moreover, as $(F,\ksi_G,\ksi_S)$ is a valid triple for $P$, the third condition gives that for each node $w$ in the forest $F$, there is exactly one leaf $p_w$ in its subtree such that $\ksi_S(p_w) \geq h_F(w)$.
Therefore, both $p_v$ and $p_{u'}$ are well-defined.

\textbf{$F'$ is a leveled forest:}
We next show that $F'$ is a leveled forest for $P$ and that for every vertex $v'$ in $F'$, it holds that $h_{F'}(v') = h_F(v')$.
If $v$ has at least two children in $F$, then this follows immediately.
Thus, it remains to consider the case that $u$ is the only child of $v$. In that case, $v$ is a root in $F$.
To see why, assume that $v$ is not a root and therefore has a parent $w$ in $F$.
As $(F,\ksi_G,\ksi_S)$ is a valid triple, this would imply that there exists a leaf $p_{vw}$ in the subtree of $v$ with $\ksi_G(p_{vw}) \geq h_F(w) > h_F(v) = \ksi_G(q)$.
In particular, this would imply that $p_{vw}$ and $q$ are two distinct leaves in the subtree of $u$ with $\ksi_G(p_{vw}) \geq h_F(v)$ and $\ksi_G(q) \geq h_F(v)$, which contradicts the fact that $(F,\ksi_G,\ksi_S)$ is a valid triple.
Thus, if $u$ is the only child of $v$, then $v$ is indeed a root and as we obtain $F'$ from $F$ by removing both $v$ and $(u,v)$, we get that the set of leaves in $F$ and $F'$ are the same and that for each node $w$ in $F'$, the subtree of $w$ in $F$ is the same as the subtree of $w$ in $F'$.
Thus, $F$' is indeed a leveled forest for $P$ and for every vertex $v'$ in $F'$, it holds that $h_{F'}(v') = h_F(v')$.

\textbf{All properties involving $\ksi'_G$ are satisfied:}
We next verify that all properties involving $\ksi'_G$ are satisfied. 

It directly follows from the way $\ksi'_G$ is derived from $\ksi_G$ and our assumption that $\ksi_G(q) \geq 1$ that $\ksi'_G$ is a valid rank function for $P$.
Moreover, as $\ksi_G$ satisfies the separation property and $\ksi'_G(p) \leq \ksi_G(p)$ for every $p \in P$, it follows that $\ksi'_G$ also satisfies the separation property.

We next check that for every point $p \in P$, $\ksi'_G(p)$ is at most the length of the path from $p$ to its root in $F'$.
First, consider the case that $p$ is a leaf in the subtree of $u$.
As $u$ is a root in $F'$, we have to show that $\ksi'_G(p) \leq h_F(u)$.
If $p = q$, this directly follows from $\ksi'_G(q) = \ksi_G(q) - 1$ and $h_F(u) = \ksi_G(q) - 1$. If $p \neq q$, we have $\ksi'_G(p) = \ksi_G(p) \leq h_F(u)$.
The last inequality follows as $\ksi_G(q) \geq h_F(v)$, and as $(F,\ksi_G,\ksi_S)$ is a valid triple for $P$, every other point $p$ in the subtree of $u$ has to satisfy $\ksi_G(p) < h_F(v)$ and therefore $\ksi_G(p) \leq h_F(u)$. 
Thus, it remains to show that $\ksi'_G(p)$ is at most the length of the path from $p$ to its root in $F'$ if $p$ is not in the subtree of $u$. This follows from the fact that $\ksi'_G(p) = \ksi_G(p)$, the fact that the path from $p$ to its root is the same in both $F$ and $F'$ and that $\ksi_G(p)$ is at most the length of the path from $p$ to its root in $F$.

Next, we verify that for every edge $(u',v')$ in $F'$, where $v'$ is the parent of $u'$, there is exactly one leaf $p'_{u'v'}$ in the subtree of $u'$ in $F'$ such that $\ksi'_G(p'_{u'v'}) \geq h_{F'}(v')$. As $(F,\ksi_G,\ksi_S)$ is a valid triple, we get the guarantee that there is exactly one leaf $p_{u'v'}$ in the subtree of $u'$ in $F$ such that $\ksi_G(p_{u'v'}) \geq h_F(v') = h_{F'}(v')$.
Therefore, it suffices to show that for every edge $(u',v')$ in $F$' and $p \in P$, $p$ is in the subtree of $u'$ in $F$ and $\ksi_G(p) \geq h_F(v')$ if and only if $p$ is in the subtree of $u'$ in $F'$ and $\ksi'_G(p) \geq h_F(v')$.
The condition clearly holds for every $p \in P$ with $\ksi'_G(p) = \ksi_G(p)$ and that has the same path to its root in $F$ and $F'$.
Thus, it remains to consider the case that $p$ is in the subtree of $u$.
If $p = q$, it is easy to verify that both conditions hold if and only if $(u',v')$ is on the path from $q$ to $u$. If $p \neq q$ is in the subtree of $u$, then $\ksi'_G(p) = \ksi_G(p) \leq h_F(u)$ and thus it is again easy to verify that the condition holds.

It remains to verify that for a given $v'$ in $F'$ and two distinct leaves $p_1$ and $p_2$ in the subtree of $v'$ in $F'$ with $\min(\ksi'_G(p_1),\ksi'_G(p_2)) \geq h_{F'}(v')$, it holds that $d(p_1,p_2) < 2^{h_{F'}(v') + 1}$.
We have shown above that $p_1$ and $p_2$ are also two distinct leaves in the subtree of $v'$ in $F$. Moreover,

\[\min(\ksi_G(p_1),\ksi_G(p_2)) \geq \min(\ksi'_G(p_1),\ksi'_G(p_2)) \geq h_{F'}(v') = h_F(v').\]
Therefore, as $(F,\ksi_G,\ksi_S)$ is a valid triple for $P$, it follows that $d(p_1,p_2) < 2^{h_F(v') + 1} = 2^{h_{F'}(v') + 1}$, as needed.

\textbf{All properties involving $\ksi'_S$ are satisfied if $u$ is the only child of $v$:}
First, we verify that for every point $p$ in $P$, $\ksi'_S(p)$ is at most the length of the path from $p$ to its root in $F'$. If $v$ is not an ancestor of $p$ in $F$, then this follows immediately from the fact that $\ksi'_S(p) = \ksi_S(p)$, the fact that the path from $p$ to its root in $F'$ and $F$ is the same and $\ksi_S(p)$ is at most the length of the path form $p$ to its root in $F'$. 
Thus, it remains to consider the case that $p$ is an ancestor of $v$ and as $u$ is the only child of $v$, also of $u$.
It suffices to show that $\ksi'_S(p) \leq h_{F'}(u) = h_F(u)$.
If $p = p_v$, this follows by definition. If $p \neq p_v$, then $\ksi_S(p) < h_F(v)$, as otherwise $p$ and $p_v$ would be two distinct leaves in the subtree of $v$ in $F$ with $\ksi_S(p) \geq h_F(v)$ and $\ksi_S(p_v) \geq h_F(v)$, which would contradict our assumption that $(F,\ksi_G,\ksi_S)$ is a valid triple for $P$.
Therefore, we have $\ksi'_S(p) = \ksi_S(p) \leq h_F(v) -1 = h_F(u)$, as needed.

Next, we show that for every node $v'$ in the forest $F'$, there is exactly one leaf $p'_{v'}$ in its subtree in $F'$ such that $\ksi'_S(p'_{v'}) \geq h_{F'}(v') = h_F(v')$. As $(F,\ksi_G,\ksi_S)$ is a valid triple, we get the guarantee that there is exactly one leaf $p_{v'}$ in the subtree of $v'$ in $F$ such that $\ksi_S(p_{v'}) \geq h_F(v') = h_{F'}(v')$.
Therefore, it suffices to show that for every node $v'$ in $F$' and $p \in P$, $p$ is in the subtree of $v'$ in $F$ and $\ksi_S(p) \geq h_F(v')$ if and only if $p$ is in the subtree of $v'$ in $F'$ and $\ksi'_S(p) \geq h_F(v')$.
The condition clearly holds for every $p \in P$ with $\ksi'_S(p) = \ksi_S(p)$ and that has the same path to its root in $F$ and $F'$.
Thus, it remains to consider the case that $p$ is in the subtree of $u$.
If $p = p_v$, it is easy to verify that both conditions hold if and only if $v'$ is on the path from $p_v$ to $u$. If $p \neq p_v$ is in the subtree of $u$, then $\ksi'_S(p) = \ksi_S(p) \leq h_F(u)$ and thus it is again easy to verify that the condition holds. 

Finally, note that $p_v$ is the only point $p$ in $P$ satisfying $\ksi_S'(p) < \ksi_S(p)$ and $p_v$ indeed satisfies $\ksi_S(p_v) = \ksi_G(q)$ and $\ksi'_S(p_v) = \ksi_G(q) - 1$. Moreover, there is no point $p$ in $P$ satisfying $\ksi'_S(p) > \ksi_S(p)$.

\textbf{All properties involving $\ksi'_S$ are satisfied if $v$ has at least two children and $p_v$ is in the subtree of $u$: }
We first verify that for every point $p \in P$, $\ksi_S'(p)$ is at most the length of the path from $p$ to its root in $F'$.
First, consider the case that $p$ is in the subtree of $u$.
In that case, we have to show that $\ksi_S'(p) \leq h_{F'}(u) = h_F(u)$. If $p = p_v$, then this follows by definition.
If $p \neq p_v$, then $\ksi_S'(p) = \ksi_S(p) < h_{F}(u)$, as otherwise $p$ and $p_v$ would be two distinct leaves in the subtree of $u$ in $F$ with $\ksi_S(p) \geq h_F(u)$ and $\ksi_S(p_v) \geq h_F(u)$.
It remains to consider the case that $p$ is not in the subtree of $v$. In that case, the path from $p$ to its root is the same in $F$ and $F'$. Thus, the only interesting case is $p = p_{u'}$, as otherwise $\ksi'_S(p) = \ksi_S(p)$. Let $w$ be the ancestor of $p_v$ of distance $\ksi_S(p_v) \geq h_F(v)$ in $F$. As $w$ is an ancestor of $v$ and $v$ is an ancestor of $p_{u'}$, it follows that $w$ is an ancestor of $p_{u'}$. Moroever, as $F$ is a leveled forest, $p_{u'}$ has the same distance to $w$ as $p_v$. Thus, $w$ is an ancestor of $p_{u'}$ in $F'$ of distance $\ksi_S(p_v) :=\ksi'_S(p_{u'})$, as needed.

Next, we show that for every node $v'$ in the forest $F'$, there is exactly one leaf $p'_{v'}$ in its subtree in $F'$ such that $\ksi'_S(p'_{v'}) \geq h_{F'}(v') = h_F(v')$. 
First, consider the case that $v'$ is an ancestor of both $p_{v}$ and $p_{u'}$ in $F$ (and therefore of $v$) satisfying $h_F(v) \leq h_F(v') \leq \ksi_S(p_v)$.
As $(F,\ksi_G,\ksi_S)$ is a valid triple, $p_v$ is the only leaf in the subtree of $v'$ in $F$ satisfying $\ksi_S(p_v) \geq h_F(v')$. As $\ksi'_S(p) = \ksi_S(p)$ for every $p \in P \setminus \{p_v,p_{u'}\}$, this implies that if $p$ is a leaf in the subtree of $v'$ in $F'$ not equal to $p_v$ or $p_{u'}$, then $\ksi_S'(p) < h_{F'}(v')$. As $\ksi_s'(p_v) = h_F(u) < h_{F'}(v')$, it follows that $p_{u'}$ is the only leaf in the subtree of $v'$ in $F'$ satisfying $\ksi_S'(p_{u'}) = \ksi_S(p_v) \geq h_{F'}(v')$, as desired. Thus, it remains to consider the case that $v'$ is not an ancestor of both $p_v$ and $p_{u'}$ in $F$ satisfying $h_F(v) \leq h_F(v') \leq \ksi_S(p_v)$. As $(F,\ksi_G,\ksi_S)$ is a valid triple, we get the guarantee that there is exactly one leaf $p_{v'}$ in the subtree of $v'$ in $F$ such that $\ksi_S(p_{v'}) \geq h_F(v') = h_{F'}(v')$. Thus, it is sufficient to show that for every $p \in P$, $p$ is in the subtree of $v'$ in $F$ and $\ksi_S(p) \geq h_F(v')$ if and only if $p$ is in the subtree of $v'$ in $F'$ and $\ksi'_S(p) \geq h_F(v')$. It is simple to verify that this is indeed the case.

Finally, note that $p_v$ is the only point $p$ in $P$ satisfying $\ksi_S'(p) < \ksi_S(p)$ and $p_v$ indeed satisfies $\ksi'_S(p_v) = h_F(u) = \ksi_G(q) - 1$ and if $v$ is a root, then also $\ksi_S(p_v) = h_F(v) =\ksi_G(q)$. Also, $p_{u'}$ is the only point $p$ in $P$ satisfying $\ksi_S'(p_{u'}) > \ksi_S(p_{u'})$ and $p_{u'}$ indeed satisfies $\ksi_S(p_{u'}) = \ksi_G(q) - 1$ and if $v$ is a root, then also $\ksi'_S(p_{u'}) = h_F(v) =\ksi_G(q)$.

\paragraph{All properties involving $\ksi'_S$ are satisfied if $v$ has at least two children and $p_v$ is not in the subtree of $u$: }
If $p_v$ is not in the subtree of $u$, then it follows that for every leaf $p$ in the subtree of $u$, $\ksi_S(p) \leq h_F(u)$. Hence, the fact that for every $p \in P$, $\ksi'_S(p)$ is at most the distance from $p$ to its root in $F'$ directly follows from $\ksi'_S(p) = \ksi_S(p)$ for every $p \in P$ and the fact that $\ksi_S(p)$ is at most the distance from $p$ to its root in $F$.

Also, the fact that for every node $v'$ in the forest $F'$, there is exactly one leaf $p'_{v'}$ in its subtree in $F'$ such that $\ksi'_S(p'_{v'}) \geq h_{F'}(v') = h_F(v')$ directly follows from $\ksi'_S(p) = \ksi_S(p)$ for every $p \in P$ and the fact that for every node $v'$ in the forest $F$, there is exactly one leaf $p_{v'}$ in its subtree in $F'$ such that $\ksi'_S(p_{v'}) \geq h_F(v')$.

Finally, there is no point $p$ in $P$ satisfying $\ksi'_S(p) < \ksi_S(p)$ and also no point $p$ in $P$ satisfying $\ksi'_S(p) > \ksi_S(p)$.

\textbf{Finishing the proof:}
From the properties we verified above, it directly follows that $(F',
\ksi'_G,\ksi'_S)$ satisfies all the guarantees of \cref{lem:rankdecrease}, except for the fact that the distance between $q$ and its root in $F'$ is $\ksi_G(q) -1$, which directly follows from the fact that the root of $q$ in $F'$ is $u$.
\end{proof}

\paragraph{DeleteWithoutMaximality operation}
With the RankDecrease operation at hand, we can now repeatedly apply this operation to formalize the deletion of a point, i.e., the operation in the left and middle pictures of \cref{fig:forest,fig:forest2}. That is, in \cref{alg:del_without_max} we show how the leveled forest is updated after the deletion of a point before we increase the ranks of some other points in the point set. 

\begin{algorithm}
\caption{DeleteWithoutMaximality Operation}
\label{alg:del_without_max}
\textbf{Input:} $P \subseteq \mathcal{X}$ for some metric space $(\mathcal{X},d)$ \\
$(F,\ksi_G,\ksi_S)$ is a valid triple for $P$ \\
$q$ is a point in $P$ \\
\textbf{Output:} $(F',\ksi'_G,\ksi'_S)$ is a valid triple for $P \setminus \{q\}$ with $\ksi'_G(p) = \ksi_G(p)$ for every $p \in P \setminus \{q\}$.

\begin{algorithmic}[1]
\Procedure{DeleteWithoutMaximality}{$q$,$P$,$F$,$\ksi_G$,$\ksi_S$}
\State $j \leftarrow \ksi_G(q)$
\State $(F^{(j)},\ksi^{(j)}_G,\ksi^{(j)}_S) \leftarrow (F,\ksi_G,\ksi_S)$
\For{$i = j-1$ down to $0$}
 \State $(F^{(i)},\ksi^{(i)}_G,\ksi^{(i)}_S) \leftarrow \texttt{RANKDECREASE}(q,P,F^{(i+1)},\ksi_G^{(i+1)},\ksi_S^{(i+1)})$
\EndFor

\State Let $F'$ be the forest we obtain from $F^{(0)}$ by removing $q$

\State Let $\ksi'_G, \ksi'_S$ be equal to $\ksi^{(0)}_G, \ksi^{(0)}_S$ restricted to $P \setminus \{q\}$
\State \Return $(F', \ksi'_G, \ksi'_S)$
\EndProcedure

\end{algorithmic}
\end{algorithm}

\begin{lemma}[DeleteWithoutMaximality Lemma]
    \label{lem:deletewithoutmaximality}
    Let $(\mathcal{X},d)$ be a metric space, $P \subseteq \mathcal{X}$ be a point set and $(F,\ksi_G,\ksi_S)$ be a valid triple for $P$ and $q \in P$. Let $(F',\ksi'_G,\ksi'_S) \leftarrow \texttt{DELETEWITHOUTMAXIMALITY}(q,P,F,\ksi_G,\ksi_S)$ (\cref{alg:del_without_max}).
    Then, $(F',\ksi'_G,\ksi'_S)$ is valid triple for point set $P \setminus \{q\}$ with $\ksi'_G(p) = \ksi_G(p)$ for every $p \in P \setminus \{q\}$.
Moreover, for every $h \in \mathbb{N}$, the number of points $p \in P \setminus \{q\}$ satisfying $\ksi_S(p) < h$ and $\ksi'_S(p) \geq h$ is at most one. Similarly, the number of points $p \in P$ satisfying $\ksi_S(p) \geq h$ and either $q = p$ or $\ksi'_S(p) < h$ is also at most one. 
\end{lemma}
\begin{proof}
We first prove by reverse induction on $i$ that for every $i \in \{0,1,\ldots,\ksi_G(q) - 1\}$, it holds that

\begin{enumerate}
    \item $(F^{(i)},\ksi_G^{(i)},\ksi_S^{(i)})$ is a valid triple for $P$
    \item $\ksi_G^{(i)}(q) = i$ and $\ksi_G^{(i)}(p) = \ksi_G(p)$ for every $p \in P \setminus \{q\}$
    \item The distance between $q$ and its root in $F^{(i)}$ is $i$.
    \item For every $h \in \mathbb{N}$, the number of points $p \in P$ satisfying $\ksi_S(p) < h$ and $\ksi^{(i)}_S(p) \geq h$ is zero if $h \leq i$ and at most one if $h > i$.
    \item For every $h \in \mathbb{N}$, the number of points $p \in P$ satisfying $\ksi_S(p) \geq h$ and $\ksi^{(i)}_S(p) < h$ is zero if $h \leq i$ and at most one if $h > i$.
\end{enumerate}

\textbf{Base case:} We start with the base case $i = \ksi_G(q) - 1$ (assuming $\ksi_G(q) \geq 1$). Note that we obtain $(F^{(\ksi_G(q) - 1)},\ksi_G^{(\ksi_G(q) - 1)},\ksi_S^{(\ksi_G(q) - 1)})$ by calling \texttt{RANKDECREASE} with input $q,P,F,\ksi_G,\ksi_S$ and that the input satisfies the preconditions. Thus, the guarantees of \texttt{RANKDECREASE} directly give that $(F^{(\ksi_G(q) - 1)},\ksi_G^{(\ksi_G(q) - 1)},\ksi_S^{(\ksi_G(q) - 1)})$ is a valid triple for $P$. We also directly get that $\ksi_G^{(\ksi_G(q) - 1)}(q) = \ksi_G(q) - 1$ and $\ksi_G^{(\ksi_G(q) - 1)}(p) = \ksi_G(p)$ for every $p \in P \setminus \{q\}$ and also that the distance between $q$ and its root in $F^{(\ksi_G(q) - 1)}$ is $\ksi_G(q) - 1$.

Next, we have to show that for every $h \in \mathbb{N}$, the number of points $p \in P$ satisfying $\ksi_S(p) < h$ and $\ksi_S^{(\ksi_G(q) - 1)} \geq h$ is zero if $h \leq \ksi_G(q) - 1$ and at most one if $h > \ksi_G(q) - 1$. According to the guarantees of \texttt{RANKDECREASE}, there exists at most one point  $p_{inc} \in P$ with $\ksi_S^{(\ksi_G(q) - 1)}(p_{inc}) > \ksi_S(p_{inc})$ and that point satisfies $\ksi_S(p_{inc}) = \ksi_G(q) - 1$, as needed.

Similarly, we have to show that for every $h \in \mathbb{N}$, the number of points $p \in P$ satisfying $\ksi_S(p) \geq h$ and $\ksi_S^{(\ksi_G(q) - 1)} < h$ is zero if $h \leq \ksi_G(q) - 1$ and at most one if $h > \ksi_G(q) - 1$.
According to the guarantees of \texttt{RANKDECREASE}, there exists at most one point  $p_{dec} \in P$ with $\ksi_S^{(\ksi_G(q) - 1)}(p_{dec}) < \ksi_S(p_{dec})$ and that point satisfies $\ksi_S^{(\ksi_G(q) - 1)}(p_{dec}) = \ksi_G(q) - 1$, as needed. 

This finishes the base case.

\textbf{Inductive step:} For the induction step, consider some $i \in \{0,1,\ldots, \ksi_G(q) - 2\}$. We show that the conditions hold for $i$ given that they hold for $i+1$. First, note that the conditions for $i+1$ directly imply that $q,P,F^{(i+1)},\ksi_G^{(i+1)},\ksi_S^{(i+1)}$ satisfy the preconditions of $\texttt{RANKDECREASE}$. Similar as for the base case, the guarantees of \texttt{RANKDECREASE} directly give that $(F^{(i)},\ksi_G^{(i)},\ksi_S^{(i)})$ is a valid triple for $P$. We also directly get that $\ksi_G^{(i)}(q) = \ksi_G^{(i+1)}(q) - 1 = (i+1) - 1 = i$ and $\ksi_G^{(i)}(p) = \ksi_G^{(i+1)}(p) = \ksi_G(p)$ for every $p \in P \setminus \{q\}$ and also that the distance between $q$ and its root in $F^{(\ksi_G(q) - 1)}$ is $\ksi_G^{(i+1)}(q) - 1 = i$.

Next, we have to show that for every $h \in \mathbb{N}$, the number of points $p \in P$ satisfying $\ksi_S(p) < h$ and $\ksi_S^{(i)} \geq h$ is zero if $h \leq i$ and at most one if $h > i$.
This is implied by the following two guarantees:
First, the condition for $i+1$ gives that for every $h \in \mathbb{N}$, the number of points $p \in P$ satisfying $\ksi_S(p) < h$ and $\ksi_S^{(i+1)}(p) \geq h$ is zero if $h \leq i+1$ and at most one if $h > i+1$.
Second, the guarantee of \texttt{RANKDECREASE} gives that there exists at most one point $p_{inc} \in P$ with $\ksi_S^{(i)}(p_{inc}) > \ksi^{(i+1)}_S(p_{inc})$ and that point satisfies $\ksi^{(i+1)}_S(p_{inc}) = \ksi_G^{(i+1)}(q) - 1 = i$ and $\ksi^{(i)}_S(p_{inc}) = \ksi_G^{(i+1)}(q) = i+1$. Here we use the fact that the distance from $q$ to its root in $F^{(i+1)}$ is $i+1 = \ksi^{(i+1)}_G(q)$.
In more detail, first consider some $h \leq i$. There is no point $p \in P$ with $\ksi_S(p) < h$ and $\ksi_S^{(i+1)}(p) \geq h$ and also no point $p \in P$ with $\ksi_S^{(i+1)}(p) < h$ and $\ksi_S^{(i)}(p) \geq h$. Hence, there is no point $p \in P$ with $\ksi_S(p) < h$ and $\ksi_S^{(i)}(p) \geq h$, as needed.
For $h = i+1$, there is no point $p \in P$ with $\ksi_S(p) < h$ and $\ksi_S^{(i+1)}(p) \geq h$ and at most one point $p \in P$ with $\ksi_S^{(i+1)}(p) < h$ and $\ksi_S^{(i)}(p) \geq h$. Hence, there is at most one point $p \in P$ with $\ksi_S(p) < h$ and $\ksi_S^{(i)}(p) \geq h$, as needed. For $h > i+1$, there is at most one point $p \in P$ with $\ksi_S(p) < h$ and $\ksi_S^{(i+1)}(p) \geq h$ and no point $p \in P$ with $\ksi_S^{(i+1)}(p) < h$ and $\ksi_S^{(i)}(p) \geq h$. Hence, there is at most one point $p \in P$ with $\ksi_S(p) < h$ and $\ksi_S^{(i)}(p) \geq h$, as needed.

Next, we have to show that for every $h \in \mathbb{N}$, the number of points $p \in P$ satisfying $\ksi_S(p) \geq h$ and $\ksi_S^{(i)}(p) < h$ is zero if $h \leq i$ and at most one if $h > i$.
This is implied by the following two guarantees:
First, the condition for $i+1$ gives that for every $h \in \mathbb{N}$, the number of points $p \in P$ satisfying $\ksi_S(p) \geq h$ and $\ksi_S^{(i+1)}(p) < h$ is zero if $h \leq i+1$ and at most one if $h > i+1$.
Second, the guarantee of \texttt{RANKDECREASE} gives that there exists at most one point $p_{dec} \in P$ with $\ksi_S^{(i)}(p_{dec}) < \ksi^{(i+1)}_S(p_{dec})$ and that point satisfies $\ksi^{(i+1)}_S(p_{dec}) = \ksi_G^{(i+1)}(q) = i + 1$ and $\ksi^{(i)}_S(p_{dec}) = \ksi_G^{(i+1)}(q) - 1 = i$. Here we use the fact that the distance from $q$ to its root in $F^{(i+1)}$ is $i+1 = \ksi^{(i+1)}_G(q)$.
The precise reasoning is analogous to the reasoning above.
This finishes the proof of the induction.

\textbf{Finishing the proof:}
We are now ready to finish the proof. First, note that we have only proven that the conditions above hold for $i = 0$ if $\ksi_G(q) \geq 1$. However, if $\ksi_G(q) = 0$, so if the for-loop is skipped, one can directly show that all of the conditions above hold for $i=0$. 
As $(F^{(0)},\ksi_G^{(0)},\ksi_S^{(0)})$ is a valid triple for $P$ and $\ksi_G^{(0)}(q) = 0$, we get that $q$ is an isolated vertex in $F^{(0)}$. Thus, removing $q$ from $F^{(0)}$ to obtain $F'$ results in a leveled forest for $P \setminus \{q\}$. Moreover, setting $\ksi'_G$ and $\ksi'_S$ equal to $\ksi^{(0)}_G$ and $\ksi^{(0)}_S$ restricted to $P \setminus \{q\}$ together with $(F^{(0)},\ksi_G^{(0)},\ksi_S^{(0)})$ being a valid triple for $P$ then also implies that $(F',\ksi'_G,\ksi'_S)$ is a valid triple for $P \setminus \{q\}$. We also have for every $p \in P \setminus \{q\}$ that $\ksi'_G(p) = \ksi^{(0)}_G(p) = \ksi_G(p)$.
Next, we have to show that for every $h \in \mathbb{N}$, the number of points $p \in P \setminus \{q\}$ satisfying $\ksi_S(p) < h$ and $\ksi_S'(p) \geq h$ is at most one. This directly follows from the property that for every $h \in \mathbb{N}$, the number of points $p \in P$ with $\ksi_S(p) < h$ and $\ksi_S^{(0)}(p) \geq h$ is at most one and $\ksi_S'(p) = \ksi_S^{(0)}(p)$ for every $p \in P \setminus \{q\}$.
Finally, we have to show that for every $h \in \mathbb{N}$, the number of points $p \in P$ satisfying $\ksi_S(p) \geq h$ and either $q = p$ or $\ksi'_S(p) < h$ is also at most one.
This directly follows from the property that for every $h \in \mathbb{N}$, the number of points $p \in P$ satisfying $\ksi_S(p) \geq h$ and $\ksi_S^{(0)}(p) < h$ is at most one and the fact that $\ksi'_S(p) = \ksi^{(0)}_S(p)$ for every $p \in P \setminus \{q\}$.
\end{proof}

\paragraph{GroupIncrease operation}
Next, we show how the leveled forest is updated for each batch of points whose geometric rank increases from $r$ to $r+1$. This formalizes the operation of adding one new node to the forest in the right pictures of \cref{fig:forest,fig:forest2}. 

\begin{algorithm}
\caption{GroupIncrease Operation}
\label{alg:groupincrease}
\textbf{Input:} 
$P \subseteq \mathcal{X}$ for some metric space $(\mathcal{X},d)$ \\
$(F,\ksi_G,\ksi_S)$ is a valid triple for $P$ \\
$h \in \mathbb{N_0}$ \\
$Q$ is a subset of $P$ such that for every $q \in Q$, it holds that (1) $\ksi_G(q) = h$, (2) $2^{h+1} \leq d(q,q') < 2^{h+2}$ for every $q' \in Q \setminus \{q\}$ and (3) $d(q,p) \geq 2^{\min(h+1,\ksi_G(p))}$ for every $p \in P \setminus Q$. \\
\textbf{Output:} $(F',\ksi'_G,\ksi'_S)$ is a valid triple for point set $P$ with $\ksi'_G(q) = \ksi_G(q) + 1$ for every $q \in Q$ and $\ksi_G'(p) = \ksi_G(p)$ for every $p \in P \setminus Q$.
\begin{algorithmic}[1]
\Procedure{GroupIncrease}{$Q$,$h$,$P$,$F$,$\ksi_G$,$\ksi_S$} 
    \If{$Q = \emptyset$}
        \State $(F',\ksi'_G,\ksi'_S) \leftarrow (F,\ksi_G,\ksi_S)$
        \State \Return $(F',\ksi'_G,\ksi'_S)$
    \EndIf
    \State For each point $q \in Q$, let $r_q$ be the root of $q$ in $F$ (satisfying $h_F(r_q) = \ksi_G(q) = h$). 
    \State Add a new vertex $r$ to $F$ and add one edge from $r_q$ to $r$ for every $q \in Q$ to obtain $F'$.  
    \State $\ksi'_G(p) =
\begin{cases}
    \ksi_G(p), & \text{if } p \in P \setminus Q \\
    h+ 1, & \text{if } p  \in Q
\end{cases}$
    \State Let $q \in Q$ be an arbitrary point and $p^*$ the unique point in the subtree of $r_q$ satisfying $\ksi_S(p^*) = h$.
    \State $\ksi'_S(p) =
\begin{cases}
    \ksi_S(p), & \text{if } p \in P \setminus \{p^*\} \\
    h + 1, & \text{if } p = p^*
\end{cases}$\\
\Return $(F',\ksi'_G,\ksi'_S)$
\EndProcedure
\end{algorithmic}
\end{algorithm}

\begin{lemma}[GroupIncrease Lemma]
\label{lem:groupincrease}
 Let $(\mathcal{X},d)$ be a metric space, $P \subseteq \mathcal{X}$ be a point set and $(F,\ksi_G,\ksi_S)$ be a valid triple for $P$. Let $h \in \mathbb{N_0}$ and $Q \subseteq P$ be a subset such that for every $q \in Q$, the following conditions hold:

\begin{enumerate}
    \item $\ksi_G(q) = h$
    \item $2^{h+1} \leq d(q,q') < 2^{h+2}$ for every $q' \in Q \setminus \{q\}$ \\
    \item $d(q,p) \geq 2^{\min(h+1,\ksi_G(p))}$ for every $p \in P \setminus Q$
\end{enumerate}
Let $(F',\ksi_G',\ksi_S') \leftarrow \texttt{GROUPINCREASE}(Q,h,P,F,\ksi_G,\ksi_S)$ (\cref{alg:groupincrease}). Then, $(F',\ksi_G',\ksi_S')$ is a valid triple for point set $P$ with $\ksi'_G(q) = \ksi_G(q) + 1$ for every $q \in Q$ and $\ksi'_G(p) = \ksi_G(p)$ for every $p \in P \setminus Q$. Moreover, there exists at most one vertex $p^* \in P$ such that $\ksi_S(p^*) = h$, $\ksi_S'(p^*) = h+1$ and for every other point $p$ in $P$ not equal to $p^*$, it holds that $\ksi_S'(p) = \ksi_S(p)$.
\end{lemma}
\begin{proof}
The only interesting case is $Q \neq \emptyset$, which we assume from now on.
Let $q \in Q$ and $r_q$ be the root of $q$ in $F$. We first verify that indeed $h_F(r_q) = \ksi_G(q)$. First, $h_F(r_q) \geq \ksi_G(q)$ follows from our assumption that $\ksi_G(q)$ is at most the length from $q$ to its root in $F$. For the sake of contradiction, assume that $h_F(r_q) > \ksi_G(q)$. Then, let $u$ and $v$ be the ancestor of $q$ in $F$ of distance $\ksi_G(q)$ and $\ksi_G(q) + 1$, respectively. By our assumption, there exists a leaf $p_{uv}$ in the subtree of $u$ with $\ksi_G(p_{uv}) \geq h_F(v) = \ksi_G(q) + 1$. In particular, this would imply that $q$ and $p_{uv}$ are two distinct leaves in the subtree of $u$ such that $\min(\ksi_G(q),\ksi_G(p_{uv})) \geq h_F(u)$. As $(F,\ksi_G,\ksi_S)$ is a valid triple, that would imply that $d(q,p_{uv}) < 2^{h_F(u) + 1} = 2^{\ksi_G(q) + 1} = 2^{h+1}$. However, as $\ksi_G(p_{uv}) \neq h$, we have $p_{uv} \notin Q$ and one of our assumptions is that $d(q,p_{uv}) \geq 2^{\min(h+1,\ksi_G(p_{uv}))} = 2^{h+1}$, a contradiction. Therefore, it indeed holds that $h_F(r_q) = \ksi_G(q) = h$.

Thus, we obtain $F'$ from $F$ by considering some non-empty subset $R_Q$ of roots in $F$, with each root having the same height, and adding an edge from each $r_q \in R_Q$ to a newly added vertex $r$.
Together with our assumption that $F$ is a leveled forest for $P$, this implies that $F'$ is also a leveled forest for $P$.

Next, we verify that $\ksi'_G$ satisfies the separation property, namely that for  any two distinct points $p_1,p_2 \in P$, it holds that $d(p_1,p_2) \geq 2^{\min(\ksi_G'(p_1),\ksi_G'(p_2))}$. If both $p_1$ and $p_2$ are in $Q$, then $\ksi'_G(p_1) = \ksi'_G(p_2) = h+1$ and one of our assumptions is that $d(p_1,p_2) \geq 2^{h+1}$. If $p_1,p_2 \notin Q$, then $\ksi_G'(p_1) = \ksi_G(p_1)$ and $\ksi_G'(p_2) = \ksi_G(p_2)$. As we assume that $\ksi_G$ satisfies the separation property, we thus get that $d(p_1,p_2) \geq 2^{\min(\ksi_G(p_1),\ksi_G(p_2))} = 2^{\min(\ksi_G'(p_1),\ksi_G'(p_2))}$. Finally, if exactly one of $p_1$ and $p_2$ is contained in $Q$, let's say $p_1$, then $d(p_1,p_2) \geq 2^{\min(h+1,\ksi_G(p_2))}$ is one of our assumptions and thus $d(p_1,p_2) \geq 2^{\min(\ksi_G'(p_1),\ksi_G'(p_2))}$ follows from $\ksi_G'(p_1) = h+1$ and $\ksi_G'(p_2) = \ksi_G(p_2)$. Thus, the separation property indeed holds.

The property that for every $p \in P$, the values of $\ksi'_G(p)$ and $\ksi'_S(p)$ are at most the length of the path from $p$ to its root in $F'$ straightforwardly follows from how we derive $(F',\ksi'_G,\ksi'_S)$ from $(F,\ksi_G,\ksi_S)$ and our assumption that $\ksi_G(p)$ and $\ksi_S(p)$ are at most the length of the path from $p$ to its root in $F$.

We also have to verify the property that for every node $v$ in the forest $F'$, there is exactly one leaf $p'_v$ in its subtree such that $\ksi'_S(p'_v) \geq h_{F'}(v)$. If $v$ is equal to the newly added vertex $r$, then it is easy to verify that $p^*$ is the only point in the subtree of $r$ with $\ksi'_S(p^*) = h+1 \geq h_{F'}(r)$. If $v$ is not equal to $r$, then one of our assumptions is that there is exactly one leaf $p_v$ in its subtree in $F$ such that $\ksi_S(p_v) \geq h_{F}(v)$. It is easy to verify that $p_v$ is also the only leaf in the subtree of $v$ in $F'$ such that $\ksi'_S(p_v) \geq h_{F'}(v)$.

Next, we verify that for each edge $(u,v)$ in $F'$, where $v$ is the parent of $u$, there is exactly one leaf $p'_{uv}$ in the subtree of $u$ such that $\ksi'_G(p'_{uv}) \geq h_{F'}(v)$. If $(u,v) = (r_q,r)$ for some $q \in Q$, then it is easy to verify that $q$ is the only point in the subtree of $r_q$ with $\ksi'_G(q) = h+1 \geq h_{F'}(r)$. Otherwise, if $(u,v)$ is also in $F$, then one of our assumptions is that there is exactly one leaf $p_{uv}$ in the subtree of $u$ in $F$ such that $\ksi_G(p_{uv}) \geq h_F(v)$. It is easy to verify that $p_{uv}$ is also the only leaf in the subtree of $u$ in $F'$ such that $\ksi'_G(p_{uv}) \geq h_{F'}(v)$.

Let $v$ be a node in $F'$, and $p_1,p_2$ be two distinct leaves in the subtree of $v$ in $F$ such that $\min(\ksi'_G(p_1),\ksi'_G(p_2)) \geq h_{F'}(v)$. We have to show that this implies $d(p_1,p_2) < 2^{h_{F'}(v) + 1}$. First, consider the case $v = r$. In that case, both $p_1,p_2 \in Q$ and one of our assumptions is that $d(p_1,p_2) < 2^{h+2} = 2^{h_{F'}(v) + 1}$. Thus, it remains to consider the case that $v \neq r$. Then, $\min(\ksi'_G(p_1),\ksi'_G(p_2)) \geq h_{F'}(v)$ implies $\min(\ksi_G(p_1),\ksi_G(p_2)) \geq h_{F'}(v) = h_F(v)$ and thus one of our assumptions is that $d(p_1,p_2) < 2^{h_F(v) + 1} = 2^{h_{F'}(v) + 1}$, as needed.

Finally, note that the fact that there exists at most one vertex $p^* \in P$ such that $\ksi_S(p^*) = h$, $\ksi'_S(p^*) = h+1$ and for every other point $p$ in $P$ not equal to $p^*$, it holds that $\ksi'_S(p) = \ksi_S(p)$ directly follows from the way $\ksi'_S$ is defined.
\end{proof}

\paragraph{Delete operation}
Finally, we put \cref{alg:del_without_max,alg:groupincrease} together to get \cref{alg:delete} that we run whenever a point is deleted from our solution. This algorithm formalizes the algorithm sketched in pictures in \cref{fig:forest,fig:forest2}. 

\begin{algorithm}
\caption{Delete Operation}
\label{alg:delete}
\textbf{Input:} $P \subseteq \mathcal{X}$ for some metric space $(\mathcal{X},d)$ and the minimum pairwise distance in $P$ is at least $1$ \\
                $(F,\ksi_G,\ksi_S)$ is a valid triple for $P$ and $\ksi_G$ additionally satisfies the maximality property \\
                $q$ is an arbitrary point in $P$ \\
                $\Delta \in \mathbb{N}$ is an upper bound on the maximum pairwise distance between any two points in $P$ \\
                There exists a point $p \in P$ with $\ksi_G(p) \geq \lceil \log(\Delta) + 1\rceil$ \\
\textbf{Output:} $(F', \ksi'_G, \ksi'_S)$ is a valid triple for point set $P \setminus \{q\}$, $\ksi'_G$ satisfies the maximality property, if $P \setminus\{q\} \neq \emptyset$, then there exists $p \in P \setminus \{q\}$ with $\ksi'_G(p) \geq \lceil\log(\Delta) + 1\rceil$.
\begin{algorithmic}[1]
\Procedure{Delete}{$q$,$P$,$F$,$\ksi_G$,$\ksi_S$,$\Delta$}
\State $(F^{(0)},\ksi^{(0)}_G,\ksi^{(0)}_S) \leftarrow \texttt{DELETEWITHOUTMAXIMALITY}(q,P,F,\ksi_G,\ksi_S)$
\For{$i = 0,1,\ldots,\lceil \log(\Delta) + 1\rceil - 1$}
\State Let $Q_i \subseteq P \setminus \{q\}$ be a \emph{maximal} subset such that for every $q_i \in Q_i$, it holds that (1) $\ksi^{(i)}_G(q_i) = i$, (2) $d(q_i,q_i') \geq 2^{i+1}$ for every $q_i' \in Q_i \setminus \{q_i\}$ and (3) $d(q_i,p) \geq 2^{\min(i+1,\ksi_G^{(i)}(p))}$ for every $p \in P \setminus (\{q\} \cup Q_i)$. 
\State $(F^{(i+1)},\ksi^{(i+1)}_G,\ksi^{(i+1)}_S) \leftarrow \texttt{GROUPINCREASE}(Q_i,i, P \setminus \{q\},F^{(i)},\ksi^{(i)}_G,\ksi^{(i)}_S)$
\EndFor
$(F',\ksi_G',\ksi_S') \leftarrow (F^{(\lceil \log(\Delta) + 1\rceil)},\ksi_G^{(\lceil \log(\Delta) + 1\rceil)},\ksi_S^{(\lceil \log(\Delta) + 1\rceil)})$
\State \Return $(F', \ksi'_G, \ksi'_S)$
\EndProcedure

\end{algorithmic}
\end{algorithm}

\begin{lemma}
[Deletion Lemma]
\label{lem:deletion}
Let $(\mathcal{X},d)$ be a metric space and $P \subseteq \mathcal{X}$ be a point set with minimum pairwise distance at least $1$. Let $(F,\ksi_G,\ksi_S)$ be a valid triple for point set $P$ and $\ksi_G$ additionally satisfies the maximality property. Let $q \in P$ and $\Delta \in \mathbb{N}$ be an upper bound on the maximum pairwise distance between any two points in $P$.
Let $(F',\ksi'_G,\ksi'_S) \leftarrow \texttt{DELETE}(q,P,F,\ksi_G,\ksi_S,\Delta)$ (\cref{alg:delete}). Then, $(F',\ksi'_G,\ksi'_S)$ is a valid triple for point set $P \setminus \{q\}$ and $\ksi'_G$ satisfies the maximality property. Moreover, if $P \setminus \{q\} \neq \emptyset$, then there exists $p \in P \setminus \{q\}$ with $\ksi'_G(p) \geq \lceil \log(\Delta) + 1\rceil$.
Moreover, for every $h \in \mathbb{N}$, there exists at most one point $p \in P$ with $\ksi_S(p) \geq h$ and either $p = q$ or $\ksi'_S(p) < h$. Also, the number of points $p \in P \setminus \{q\}$ satisfying $\ksi_S(p) < h$ and $\ksi'_S(p) \geq h$ is at most $2$.
\end{lemma}
\begin{proof}

We first show by induction that for every $i \in \{0,1,\ldots, \lceil \log(\Delta) + 1\rceil\}$, the following holds:

\begin{enumerate}
    \item $(F^{(i)},\ksi_G^{(i)},\ksi_S^{(i)})$ is a valid triple for $P \setminus \{q\}$.
    \item For every $p \in P \setminus \{q\}$ with $\ksi^{(i)}_G(p) < i$, there exists a $p' \in P \setminus \{q\}$ with $\ksi_G^{(i)}(p') > \ksi_G^{(i)}(p)$ and $d(p,p') < 2^{\ksi_G^{(i)}(p) + 1}$.
    \item For every $p \in P \setminus \{q\}$, it holds that $\ksi^{(i)}_G(p) \geq \ksi_G(p)$. \\
    \item For every $p \in P \setminus \{q\}$, it holds that $\ksi^{(i)}_S(p) \geq \ksi^{(0)}_S(p)$.
    \item For every $j \geq i$ and $p \in P \setminus \{q\}$ with $\ksi_S^{(0)}(p) \leq j$, we have $\ksi_S^{(i)}(p) \leq j$.  
    \item For every $j \leq i$, there exists at most one $p \in P \setminus \{q\}$ with $\ksi^{(0)}_S(p) < j$ and $\ksi^{(i)}_S(p) \geq j$. 
\end{enumerate}

\textbf{Base case:} First, note that the preconditions for the \texttt{DELETE}-operation directly imply that the preconditions are satisfied when we call the \texttt{DELETEWITHOUTMAXIMALITY}-operation. Thus, it directly follows from the guarantees of the \texttt{DELETEWITHOUTMAXIMALITY}-operation that $(F^{(0)},\ksi_G^{(0)},\ksi_S^{(0)})$ is a valid triple for $P \setminus \{q\}$ and we also get the guarantee that $\ksi^{(0)}_G(p) = \ksi_G(p)$ for every $p \in P \setminus \{q\}$. All other conditions are trivially satisfied for $i = 0$.

\textbf{Inductive step:} Now, assume that the conditions hold for some $i \in \{0,1,\ldots,\lceil \log(\Delta) + 1\rceil - 1\}$. We next show that then the conditions also hold for $i+1$.

We first have to verify that the preconditions are met when we call $\texttt{GROUPINCREASE}$ with input $Q_i,i,P \setminus \{q\},F^{(i)}, \ksi_G^{(i)}$ and $\ksi_S^{(i)}$. The only precondition that does not directly follow from our assumption that $(F^{(i)},\ksi_G^{(i)},\ksi_S^{(i)})$ is a valid triple and the conditions we impose on $Q_i$ is the property that $d(q_i,q'_i) < 2^{i+2}$ for every $q_i,q'_i \in Q_i$. By triangle inequality, it suffices to show that $d(q_i,q) < 2^{i+1}$ for every $q_i \in Q_i$. Note that we assume that $\ksi_G$ satisfies the maximality property and that there exists a point $p \in P$ with $\ksi_G(p) \geq \lceil \log(\Delta) + 1\rceil$. As $q_i \in Q_i$, we have $\ksi_G(q_i) \leq \ksi^{(i)}_G(q_i) = i < \lceil \log(\Delta) + 1\rceil$. Hence, the maximality property implies that there exists $p' \in P$ with $\ksi_G(p') > \ksi_G(q_i)$ and $d(q_i,p') < 2^{\ksi_G(q_i) + 1} \leq 2^{\min(i+1,\ksi_G(p'))}$. The second condition of $Q_i$ implies that $p' \notin Q_i$ and the third condition together with $\ksi_G^{(i)}(p) \geq \ksi_G(p)$ for every $p \in P \setminus \{q\}$ implies that $p' \notin P \setminus (\{q\} \cup Q_i)$. Thus, $p' = q$ and we indeed get that $d(q_i,q) < 2^{i+1}$. Thus, the preconditions are satisfied when we call $\texttt{GROUPINCREASE}$ with input $Q_i,i,P \setminus \{q\},F^{(i)}, \ksi_G^{(i)}$ and $\ksi_S^{(i)}$.

Thus, we directly get that $(F^{(i+1)},\ksi_G^{(i+1)},\ksi_S^{(i+1)})$ is a valid triple for $P \setminus \{q\}$.

Next, let $p \in P \setminus \{q\}$ with $\ksi_G^{(i+1)}(p) < i+1$. We have to show that there exists $p' \in P \setminus \{q\}$ with $\ksi_G^{(i+1)}(p') > \ksi_G^{(i+1)}(p)$ and $d(p,p') < 2^{\ksi_G^{(i+1)}(p) + 1}$. As $\ksi_G^{(i+1)}(p) < i+1$, we have $p \notin Q_i$ and therefore $\ksi_G^{(i+1)}(p) = \ksi_G^{(i)}(p)$. First, consider the case that $\ksi_G^{(i)}(p) < i$. Then, the induction hypothesis implies that there exists $p' \in P \setminus \{q\}$ with $\ksi_G^{(i+1)}(p') \geq \ksi_G^{(i)}(p') > \ksi_G^{(i)}(p) = \ksi_G^{(i+1)}(p)$ and $d(p,p') < 2^{\ksi_G^{(i)}(p) + 1} = 2^{\ksi_G^{(i+1)}(p) + 1}$. Thus, it remains to consider the case that $\ksi_G^{(i+1)}(p) = \ksi_G^{(i)}(p) = i$. As $Q_i$ is a \emph{maximal} subset with the specified properties, this implies that there exists $q_i \in Q_i$ with $d(p,q_i) < 2^{i+1}$ or there exists some other point $p' \in P \setminus (\{q\} \cup Q_i)$ with $d(p,p') < 2^{\min(i+1,\ksi_G^{(i)}(p'))}$. In the former case, we have $\ksi_G^{(i+1)}(q_i) = i+1 > \ksi_G^{(i+1)}(p)$ and $d(p,q_i) < 2^{\ksi_G^{(i+1)}(p) + 1}$, as needed. Thus, it remains to consider the latter case. As $\ksi_G^{(i+1)}$ satisfies the separation property, we have $d(p,p') \geq 2^{\min(i,\ksi_G^{(i+1)}(p'))}$. Together with $d(p,p') < 2^{\min(i+1,\ksi_G^{(i)}(p'))}$, this implies that $\ksi_G^{(i+1)}(p') \geq i+1 > \ksi_G^{(i+1)}(p)$, which together with $d(p,p') < 2^{\min(i+1,\ksi_G^{(i)}(p'))} \leq 2^{\ksi_G^{(i+1)}(p) + 1}$ shows that the second property is satisfied.

The third property, namely that for every $p \in P \setminus \{q\}$, it holds that $\ksi_G^{(i+1)}(p) \geq \ksi_G(p)$ directly follows from $\ksi_G^{(i)}(p) \geq \ksi_G(p)$ for every $p \in P \setminus \{q\}$ and $\ksi_G^{(i+1)}(p) \geq \ksi^{(i)}_G(p)$ for every $p \in P \setminus \{q\}$.

Note that $\texttt{GROUPINCREASE}$ guarantees that there exists at most one $p^* \in P \setminus \{q\}$ with $\ksi^{(i)}_S(p^*) = i$ and $\ksi^{(i+1)}_S(p^*) = i+1$ and that for every $p \in P \setminus \{q\}$ not equal to $p^*$, we have $\ksi^{(i+1)}_S(p) = \ksi^{(i)}_S(p)$. Using this guarantee, a simple inductive argument shows that properties $4$, $5$ and $6$ are preserved. This finishes the induction.

\textbf{Finishing the proof:}
We are now ready to show that $(F',\ksi'_G,\ksi'_S)$ satisfies the output guarantees. The fact that $(F',\ksi_G',\ksi_S')$ is a valid triple for point set $P \setminus \{q\}$ directly follows from the fact that $(F^{(\lceil \log(\Delta) + 1\rceil)},$ $ \ksi_G^{(\lceil \log(\Delta) + 1\rceil)}, \ksi_S^{(\lceil \log(\Delta) + 1\rceil)})$ is a valid triple for point set $P \setminus \{q\}$. 

The second property for $i = \lceil \log(\Delta + 1)\rceil$ together with $\ksi'_G = \ksi_G^{(\lceil \log(\Delta) + 1\rceil)}$ implies that for every $p \in P \setminus \{q\}$ with $\ksi'_G(p) < \lceil \log(\Delta) + 1\rceil$, there exists a $p' \in P \setminus \{q\}$ with $\ksi_G'(p') > \ksi_G'(p)$ and $d(p,p') < 2^{\ksi_G'(p) + 1}$.
First, this implies that if $P \setminus \{q\} \neq \emptyset$, then there exists $p \in P \setminus \{q\}$ with $\ksi'_G(p) \geq \lceil \log(\Delta) + 1 \rceil$. Moreover, it also implies that $\ksi'_G$ satisfies the maximality property, using the fact that there exists at most one point $p \in P \setminus \{q\}$ with $\ksi'_G(p) \geq \lceil \log(\Delta) + 1 \rceil$. This in turn follows as for any two distinct points $p_1,p_2 \in P \setminus \{q\}$, we have $2^{\min(\ksi'_G(p_1),\ksi'_G(p_2))} \leq d(p_1,p_2) \leq \Delta$, where the first inequality follows because $\ksi'_G$ satisfies the separation property.

Next, consider an arbitrary $h \in \mathbb{N}$. We first have to show that there exists at most one $p \in P$ with $\ksi_S(p) \geq h$ and either $p = q$ or $\ksi'_S(p) < h$.
This directly follows from the guarantee of \texttt{DELETEWITHOUTMAXIMALITY} that there exists at most one $p \in P$ with $\ksi_S(p) \geq h$ and either $p = q$ or $\ksi_S^{(0)}(p) < h$ together with the guarantee that $\ksi'_S(p) \geq \ksi_S^{(0)}(p)$ for every $p \in P \setminus \{q\}$.

Finally, it remains to show that the number of points $p \in P \setminus \{q\}$ satisfying $\ksi_S(p) < h$ and $\ksi'_S(p) \geq h$ is at most $2$. This follows from the guarantee of \texttt{DELETEWITHOUTMAXIMALITY} that there exists at most one $p \in P \setminus \{q\}$ satisfying $\ksi_S(p) < h$ and $\ksi_S^{(0)}(p) \geq h$ and the guarantee that there exists at most one $p \in P \setminus \{q\}$ satisfying $\ksi_S^{(0)}(p) < h$ and $\ksi'_S(p) \geq h$. This finishes the proof. 
\end{proof}

\subsection{Finishing the proof of our main theorem}
\label{sec:finish}

Finally, we can finish the proof of our main theorem that we restate here for convenience. 

\maintechnical*

\begin{proof}
    For $i \in \{1,2,\ldots,n-1\}$, we denote by $q_i$ the unique point in $P_i \triangle P_{i+1}$.
    First, note that we can assume without loss of generality that $P_1 = \emptyset$.
    We associate a triple $(F^{(i)},\ksi^{(i)}_G,\ksi^{(i)}_S)$ with every $i \in \{1,2,\ldots,n\}$. We define $(F^{(1)},\ksi^{(1)}_G,\ksi^{(1)}_S)$ as the unique valid triple for the point set $P_1 = \emptyset$.

    For $i \in \{1,\ldots,n-1\}$, we define $(F^{(i+1)},\ksi^{(i+1)}_G,\ksi^{(i+1)}_S)$ inductively by setting $(F^{(i+1)},\ksi^{(i+1)}_G,\ksi^{(i+1)}_S) \leftarrow \texttt{INSERT}(q,P_i,F^{(i)},\ksi_G^{(i)},\ksi_S^{(i)},\Delta)$ if $q_i \notin P_i$ and $(F^{(i+1)},\ksi^{(i+1)}_G,\ksi^{(i+1)}_S) \leftarrow \texttt{DELETE}(q,P_i,F^{(i)},\ksi_G^{(i)},\ksi_S^{(i)},\Delta)$ otherwise.  

    We first show by induction that for every $i \in \{1,2,\ldots,n\}$, it holds that $(F^{(i)},\ksi^{(i)}_G,\ksi^{(i)}_S)$ is a valid triple for $P_i$, $\ksi^{(i)}_G$ satisfies the maximality property and if $P_i \neq \emptyset$, then there exists a $p \in P_i$ with $\ksi_G(p) \geq \lceil \log(\Delta) + 1 \rceil$.

    The base case $i = 0$ trivially holds. Now, consider some arbitrary $i \in \{1,2,\ldots,n-1\}$ and assume the statement holds for $i$. We show that the statement also holds for $i+1$.

    First, consider the case that $q_i \notin P_i$. Then, $q,P_i,F^{(i)}.\ksi_G^{(i)},\ksi_S^{(i)},\Delta$ satisfy the input conditions of \cref{lem:insertion}. Thus, $\cref{lem:insertion}$ guarantees that $(F^{(i+1)},\ksi^{(i+1)}_G,\ksi^{(i+1)}_S)$ satisfies the conditions for $i+1$.

    Similarly, if $q_i \in P_i$, then $q,P_i,F^{(i)}.\ksi_G^{(i)},\ksi_S^{(i)},\Delta$ satisfy the input conditions of \cref{lem:deletion}. Thus, \cref{lem:deletion} guarantees that $(F^{(i+1)},\ksi^{(i+1)}_G,\ksi^{(i+1)}_S)$ satisfies the conditions for $i+1$, which finishes the induction.

    For $i \in \{1,2,\ldots,n\}$, if $|P_i| \leq k$, we define $C_i = P_i$. Otherwise, if $|P_i| > k$, we include in $C_i$ the $k$ points in $P_i$ that got assigned the largest rank by $\ksi^{(i)}_S$. Ties can be resolved arbitrarily, with the exception that a point $p$ is favored to be in $C_i$ if it was already contained in $C_{i-1}$.

    If $|P_i| \leq k$, then  $\max_{p \in P_i} d(p,C_i) \leq 24 \cdot \min_{C \subseteq \mathcal{X},|C| \leq k} \max_{p \in P_i} d(p,C)$ holds trivially. If $|P_i| > k$, then $\max_{p \in P_i} d(p,C_i) \leq 24 \cdot \min_{C \subseteq \mathcal{X},|C| \leq k} \max_{p \in P_i} d(p,C)$ follows from \cref{lem:sufficient_for_good_approximation} as $C_i$ contains every point $p \in P_i$ with $\ksi_S^{(i)}(p) > (\ksi_S^{(i)})^*(k+1)$ and $(\ksi^{(i)}_G,\ksi^{(i)}_S)$ is a valid tuple for $P_i$ according to \cref{lem:valid_triple_maximality_implies_valid_tuple}.

    It remains to upper bound $|C_i \triangle C_{i+1}|$ for a given $i \in \{1,2,\ldots,n-1\}$.

    First, consider the case that $q_i \notin P_i$. \cref{lem:insertion} guarantees that $\ksi^{(i)}_S(p) = \ksi^{(i)}_S(p)$ for every $p \in P_i$. Thus, it directly follows from the way we define $C_{i+1}$ that $q_i$ can be the only point contained in $C_{i+1} \setminus C_i$ and thus there is also at most one point in $C_i \setminus C_{i+1}$. Hence, $|C_i \triangle C_{i+1}| \leq 2$.

    It remains to consider the case that $q_i \in P_i$. Then, \cref{lem:deletion} implies that for every $h \in \mathbb{N}$, there exists at most one point $p \in P_i$ with $\ksi^{(i)}_S(p) \geq h$ and either $p = q_i$ or $\ksi^{(i+1)}_S(p) < h$. Also, the number of points $p \in P_{i+1}$ satisfying $\ksi^{(i)}_S(p) < h$ and $\ksi^{(i+1)}_S(p) \geq h$ is at most $2$. It is easy to verify that this implies that  $|C_i \triangle C_{i+1}| \leq 4$.
    Finally, it is easy to verify that for every $i \in \{1,2,\ldots,n\}$, $C_i$ is fully determined by $P_1,P_2,\ldots,P_i,(\mathcal{X},d),k$ and $\Delta$, which finishes the proof.

\end{proof}

\bibliographystyle{alpha}
\bibliography{cluster}

\end{document}